\documentclass{SCIS2024}
\usepackage{graphicx}
\usepackage{svg}
\usepackage{subfig}

\usepackage{amsmath,amsfonts}
\usepackage{amssymb}
\usepackage{algorithmic}
\usepackage{algorithm}

\begin{document}
\ArticleType{RESEARCH PAPER}
\Year{2024}
\Month{}
\Vol{}
\No{}
\DOI{}
\ArtNo{}
\ReceiveDate{}
\ReviseDate{}
\AcceptDate{}
\OnlineDate{}

\title{Diffusion-based Auction Mechanism for Efficient Resource Management in 6G-enabled Vehicular Metaverses
}{Title keyword 5 for citation Title for citation Title for citation}

\author[1]{Jiawen Kang}{{kjwx886@163.com}}
\author[1]{Yongju Tong}{}
\author[1]{Yue Zhong}{}
\author[1]{Junlong Chen}{}
\author[2]{Minrui Xu}{}
\author[2]{\\Dusit Niyato}{}
\author[1]{Runrong Deng}{}
\author[3]{Shiwen Mao}{}

\AuthorMark{Kang J W}

\address[1]{School of Automation, Guangdong University of Technology, Guangzhou {\rm 510006}, China}
\address[2]{School of Computer Science and Engineering, Nanyang Technological University, Singapore {\rm 639798}, Singapore}
\address[3]{Department of Electrical and Computer Engineering, Auburn University, Auburn {\rm AL 36849-5201}, USA}

\abstract{
The rise of 6G-enable Vehicular Metaverses is transforming the automotive industry by integrating immersive, real-time vehicular services through ultra-low latency and high bandwidth connectivity. In 6G-enable Vehicular Metaverses, vehicles are represented by Vehicle Twins (VTs), which serve as digital replicas of physical vehicles to support real-time vehicular applications such as large Artificial Intelligence (AI) model-based Augmented Reality (AR) navigation, called VT tasks. VT tasks are resource-intensive and need to be offloaded to ground Base Stations (BSs) for fast processing. However, high demand for VT tasks and limited resources of ground BSs, pose significant resource allocation challenges, particularly in densely populated urban areas like intersections. As a promising solution, Unmanned Aerial Vehicles (UAVs) act as aerial edge servers to dynamically assist ground BSs in handling VT tasks, relieving resource pressure on ground BSs. However, due to high mobility of UAVs, there exists information asymmetry regarding VT task demands between UAVs and ground BSs, resulting in inefficient resource allocation of UAVs. To address these challenges, we propose a learning-based Modified Second-Bid (MSB) auction mechanism to optimize resource allocation between ground BSs and UAVs by accounting for VT task latency and accuracy. Moreover, we design a diffusion-based reinforcement learning algorithm to optimize the price scaling factor, maximizing the total surplus of resource providers and minimizing VT task latency. Finally, simulation results demonstrate that the proposed diffusion-based MSB auction outperforms traditional baselines, providing better resource distribution and enhanced service quality for vehicular users.}

\maketitle

\section{Introduction}

With the rapid development of Sixth Generation (6G) technology, the concept of 6G-enabled Vehicular Metaverses is revolutionizing intelligent transportation systems by offering ultra-reliable, low-latency communication, massive connectivity, and advanced network capabilities that seamlessly link the physical and virtual worlds\cite{zhong2023blockchain}. 
Vehicle Twins (VTs) are digital replicas of physical vehicles within 3D virtual environments, offering accurate lifecycle representations and managing vehicular applications like large Artificial Intelligence (AI) model-based Augmented Reality (AR) navigation \cite{tong2024diffusion}.
In 6G-enabled Vehicular Metaverses, vehicular users can receive large AI model-based vehicular services called VT tasks by constructing and updating the VTs in nearby terrestrial infrastructure like ground Base Stations (BSs). 
Given limited coverage of ground BSs and rapid movement of vehicles, a single ground BS is unable to continuously provide immersive vehicular services to vehicular users~\cite{luo2023privacy}. Consequently, to ensure an uninterrupted and immersive vehicular user experience, VT tasks must be migrated to the next ground BS as the vehicle is about to exit the coverage of the current ground BS~\cite{chen2023multiple}. Recently, integrating Vehicular Metaverses into urban environments represents a significant advancement in immersive experiences, where users can interact with digital content in real time through large AI model-based
vehicular systems \cite{10401029}.
However, as the scale of the Vehicular Metaverse expands and the number of vehicular users increases, a critical challenge emerges when VT tasks are migrated, particularly large AI model-based VT tasks, i.e., network resource limitation. This issue becomes particularly prominent in areas with high vehicular user density, such as intersections, urban zones, and commercial zones, where seamless communication and immersive Vehicular Metaverse services are required.
To address these challenges, air-ground integrated networks, composed of Unmanned Aerial Vehicles (UAVs) and ground BSs, provide dynamic resource allocation, expanding network coverage and alleviating communication bottlenecks in densely populated zones \cite{pang2021uav,kang2024uav}. 
Therefore, the capability of air-ground integrated networks to dynamically manage resources can particularly enhance service quality in user-dense regions, offering timely responses and improving the quality of experience of vehicular users.

In air-ground integrated networks, ground BSs and UAVs serve as critical resource providers, delivering computing and communication resources to support executing VT tasks migration within 6G-enable Vehicular Metaverses \cite{dong2021uavs}. These VT task migrations often require continuous and high-quality resource allocation, especially for large AI model-based VT tasks, which is essential for maintaining an immersive experience for vehicular users. However, the dynamic nature of resource demand in vehicular environments presents challenges, e.g., the information asymmetry between ground BSs and UAVs regarding the demand for VT task migration. Specifically, ground BSs can utilize historical data to accurately predict and assess the resource requirements of vehicular users with their fixed infrastructures and stable operations \cite{noor2020survey}. Although UAVs can offer flexible and dynamic coverages, they face challenges in collecting real-time feedback due to their mobility. Therefore, this mobility results in delayed information gathering (i.e., the demand for VT task migration), leading to less precise estimations of resource valuation compared to ground BSs \cite{mozaffari2019tutorial}.

To effectively address the resource allocation challenges in air-ground integrated networks, where ground BSs and UAVs serve as critical resource providers for vehicular users executing VT tasks in Vehicular Metaverses, it is essential to design mechanisms that mitigate the information asymmetry between these resource providers \cite{zhu2022traffic}. Traditional auction mechanisms, such as first-price or second-price auctions, often lead to inefficient resource allocation due to strategic bidding and adverse selection, which is not suitable for solving the resource allocation challenges in air-ground integration networks. The Modified Second Bid (MSB) auction offers significant advantages by incorporating a price scaling factor that ensures more efficient resource pricing and allocation. As introduced in \cite{arnosti2016adverse}, the MSB auction is strategy-proof and effectively eliminates adverse selection, making it a robust solution in environments characterized by information asymmetry \cite{xu2024cached}. 
Therefore, we propose a new MSB auction mechanism that evaluates resource allocation based on two key metrics, i.e., the latency of processing VT tasks and the precision of these VT task outcomes (e.g., resolution or pixel accuracy), which represent the common value and match value of resource provider, respectively.

The rise of Generative Artificial Intelligence (GAI) presents transformative potential, surpassing conventional AI paradigms by incorporating a wide range of models and methodologies, e.g., Transformer, Generative Adversarial Networks (GANs), and Generative Diffusion Models (GDMs) \cite{du2024enhancing}. Among these, GDMs are particularly noteworthy due to their unique approach to data generation and their ability to model complex data distributions with high precision \cite{10419041}. By incorporating the diffusion model and Reinforcement Learning (RL), we propose a Diffusion-based RL algorithm, which is a generative approach that excels at learning and representing intricate environments through noise injection and denoising processes \cite{du2024enhancing}. In this paper, we employ the Diffusion-based RL algorithm to solve the MSB auction problem, optimizing the price scaling factor to maximize the surplus of resource providers and minimize the latency of VT tasks. The main contributions of this paper are summarized as follows. 
 \begin{itemize}
     \item We propose a new resource allocation framework to address information asymmetry between UAVs and ground BSs due to the demand for VT task migration, providing an immersive experience for vehicular users in 6G-enabled Vehicular Metaverses.
 \end{itemize}
\begin{itemize}
    \item In this framework, we design a new MSB auction, considering the VT task latency and the accuracy of these VT tasks (i.e., pixel points or resolution). Additionally, we prove that the proposed MSB auction is strategy-proof and effectively eliminates adverse selection, making it a robust solution in environments characterized by information asymmetry.
\end{itemize}
\begin{itemize}
    \item Unlike traditional DRL, we propose a Diffusion-based RL algorithm that has powerful learning and characterization capabilities, applying the MSB auction system to maximize the total surplus of resource providers and minimize the average latency of VT tasks.
\end{itemize}

The rest of this paper is organized as follows. In Section \ref{Related Work}, we provide a review of related works. Section \ref{System Model} presents the system model for handling large AI model-based VT tasks within 6G-enabled Vehicular Metaverses. In Section \ref{Market Design}, we formulate the problem and propose the market design. In Section \ref{Diffusion-based}, we introduce the proposed Diffusion-based MSB auction. Section \ref{experiment} discusses the simulation setup and results, followed by the conclusion in Section \ref{conclusion}.
\section{Related Works}\label{Related Work}
\subsection{Vehicular Metaverses}
The term ``metaverse" was first created by Neal Stephenson in the science fiction novel named \textit{Snow Crash} in 1992 \cite{stephenson1994snow}. By establishing an interconnected and immersive universe, the metaverse establishes a framework that fosters real-time communication and dynamic interactions with digital artifacts \cite{mystakidis2022metaverse}. The metaverse was derived from the fusion of the prefix ``meta" (signifying transcendence), and the suffix ``verse" (abbreviated form of the universe), which denotes a computer-generated realm that parallels the real world \cite{wang2022survey}. Currently, utilizing the metaverse finds implementation across various industries, notably within the automotive sector. Amidst the swift advancement of technologies encompassing AI, blockchain, and 6G, the attainment of an immersive Vehicular Metaverses experience emerges as a pivotal trend in the prospective evolution of the intelligent connected automobile industry \cite{10401029}. The authors in \cite{zhou2022vetaverse} introduced the term ``Vetaverse", which amalgamates ``vehicle" and ``metaverse", symbolizing the convergence of the automotive industry and the metaverse domain. In transportation systems, Vehicular Metaverses can provide fully immersive and hyper-realistic metaverse services to metaverse users inside the vehicles through ground BSs, e.g., passengers can not only observe the scenery outside the vehicles but also see AR recommendations covering real-world landmarks and landscapes through the Head-Up Displays (HUDs) \cite{xu2023generative}. The Vehicular Metaverse is delineated as an immersive integration of vehicular communication that amalgamates the virtual space with real-world data \cite{xu2023epvisa}. It is esteemed as an emergent paradigm entailing profound fusion between the Internet of Vehicles and the metaverse domain, thereby furnishing metaverse users with nascent vehicular virtual services. 

The physical-virtual synchronization system of the Vehicular Metaverse involves two phases. In the Physical-to-Virtual (P2V) phase, vehicles use sensors to perceive the environment and update the Digital Twins (DTs) in the virtual world with historical sensor and passenger biometric data \cite{xu2023epvisa}. In the Virtual-to-Physical (V2P) phase, the resource providers provide metaverse services to metaverse users through DTs, bridging the virtual and physical realms \cite{xu2023epvisa}. DTs are virtual counterparts that faithfully represent the complete life cycle of physical objects in a virtual environment \cite{juarez2021digital}. Similarly, VTs serve as comprehensive digital counterparts, encompassing the lifespan of different immersive metaverse services such as large AI model-based AR navigation, which can meet the requirements of metaverse users. The dynamic physical world of vehicles encompasses critical information and attributes of tangible entities, requiring continuous updates to the real-world characteristics of VTs within the virtual environment \cite{du2023yolo}. 
\subsection{Resource Allocation in 6G-enabled Vehicular Metaverses}
Providing services to metaverse users by constructing and updating VTs requires a lot of resources, including computing resources to perform intensive tasks, storage resources, and bandwidth resources necessary to maintain ultra-high-speed and low-latency connections \cite{alkhoori2024latency}. Consequently, addressing resource allocation challenges is crucial to incentivize resource providers in 6G-enabled Vehicular Metaverses. 
In \cite{zhong2023blockchain}, the authors presented a framework that leverages both blockchain technology and game theory to optimize the allocation of resources for VTs migration in vehicular networks. In \cite{wen2023task}, the authors proposed an incentive mechanism using the Age of Migration Task theory to ensure efficient bandwidth distribution for VTs migration in the Vehicular Metaverses. Although both studies provide theories such as game theory and contract theory to solve the resource allocation problem in Vehicular Metaverses, they do not consider the resource allocation in 6G-enabled Vehicular Metaverses.
In \cite{10373683}, the authors presented a novel approach for efficiently supporting metaverse streaming over 6G mobile networks by employing Neyman-Pearson criterion-driven Network Functions Virtualization (NFV) and Software-Defined Networking (SDN) architectures, to address the challenges of delay and bandwidth demands. However, this work primarily focuses on network optimization rather than a market-based resource allocation mechanism. 
Auction mechanisms have been widely recognized for their advantages in efficiently allocating resources, particularly in dynamic and competitive environments.
The authors in \cite{9838736} proposed a novel learning-based incentive mechanism framework aimed at evaluating and enhancing the immersive experience of Virtual Reality (VR) Vehicular Metaverse Users (VMUs) in the wireless edge-empowered metaverse. They introduced a utilization of a double Dutch auction mechanism for matching and pricing VR services between VR VMUs and VR service providers. In \cite{10225312}, the authors designed an auction mechanism grounded in the Myerson theorem and deep learning to facilitate edge computing transactions between virtual service providers and an edge computing provider. However, these existing auction mechanisms cannot fully address the information asymmetry between ground BSs and UAVs in 6G-enabled Vehicular Metaverses due to challenges including strategic bidding and adverse selection. To address these limitations, we propose a new MSB auction, which incorporates a price scaling factor to ensure more efficient resource pricing and allocation, effectively mitigating strategic bidding and adverse selection.
\subsection{Generative Models in Auction Mechanisms}
GAI introduces transformative potential, surpassing conventional AI paradigms by integrating various models and methodologies, such as Transformers, GANs, and GDMs~\cite{du2024enhancing}. Among these, GDMs stand out for their unique data generation process and their ability to model complex data distributions with high precision \cite{10419041}. GDMs work by adding noise to data and then gradually denoising it, making them highly effective in modeling intricate and dynamic environments~\cite{san2021noise}. This step-by-step approach allows GDMs to capture uncertainties and long-term dependencies, particularly suited to auctions, i.e., predicting bidder behavior and optimizing auction strategies in highly competitive and unpredictable environments. In \cite{tong2024multi}, the authors introduced an attribute-aware auction-based mechanism for optimizing resource allocation during the migration of VTs in Vehicular Metaverses, utilizing a two-stage matching process and a GPT-based DRL auctioneer. In\cite{10.1145/3637528.3671526}, the authors proposed DiffBid, a conditional diffusion modeling approach for auto-bidding in online advertising, addresses the limitations of traditional systems based on Markov Decision Processes (MDPs), which struggle in dynamic, long-term environments. DiffBid models the entire bid trajectory and proves its effectiveness with experiments, showing the strong ability of GDM to generate more accurate bids in a dynamic environment. Despite these advancements, the application of Diffusion-based MSB auctions has to be explored to address resource allocation challenges (i.e., the information asymmetry between UAVs and ground BSs due to dynamic demands of VT tasks) in 6G-enabled Vehicular Metaverses. This gap presents an opportunity to leverage Diffusion-based MSB auctions to enhance the efficiency of resource distribution within 6G-enabled Vehicular Metaverses. 
\section{System Model}\label{System Model}
The system comprises $N+1$ computing and communication resource providers, including one or several UAVs and multiple ground BSs equipped with edge servers, as shown in Fig.~\ref{framework_system}. The resource providers are represented by the set $\mathcal{N}=\{0,1,\cdots,N\}$, where 0 represents the UAV and $\{1,2,\cdots,N\}$ represents the set of ground BSs. Each ground BS has different computing resources (e.g., CPUs and GPUs). The GPU computational efficiencies of ground BSs are denoted by the set $G=\{f_1^g,\cdots,f_n^g,\cdots,f_N^g\}$, and their CPU capabilities are represented by the set $C=\{f_1^c,\cdots,f_n^c,\cdots,f_N^c\}$. The available bandwidth allocated by UAV 0 and ground BS $n=\{1,2,\cdots, N\}$, is denoted as $B_0$ and $B_n$, respectively, where $B_n=\begin{pmatrix}B_n^u,B_n^d\end{pmatrix}$ consists of uplink bandwidth $B_n^u$ and downlink bandwidth $B_n^d$. 
Vehicular users can enter the 6G-enabled Vehicular Metaverse to access various real-time large AI model-based vehicular services, such as AR navigation, pedestrian detection, and 3D entertainment. The set of vehicular users, that request the processing of VT tasks, is represented by $\mathcal{M}=\{1,\cdots,m,\cdots,M\}$. Specifically, at time slot $t$, when vehicles leave the coverage of the current ground BS, vehicular users send requirements for processing VT tasks, which are handled by either a UAV or one of the ground BSs, depending on resource availability of resource providers and the Diffusion-based MSB (DMSB) auction results. 
In our system model, an accuracy of processing VT tasks required by vehicular user $m$ is considered fixed and independent of the computing resources or bandwidth used~\cite{singh2022ai}. For example, we consider a large AI model-based VT task for pedestrian detection required by vehicular user $m$. In this task, the camera of a vehicle captures frames of its surrounding environment, consisting of multiple pixel points. We denote $\Omega_m=\{w_m^1,w_m^2,\cdots,w_m^k\}$ as the visual information required to be processed by a UAV or a ground BS. The UAV or ground BS processes these pixel points, resulting in a set of processed pixel points $\Omega_n=\{c_n^1, c_n^2,\cdots, c_n^k\}$. The accuracy of the VT tasks depends on whether each processed pixel point $c_n^i$ corresponds to its original pixel $w_m^i$. We then show more detail about the network, latency, and accuracy models.
\begin{figure*}[t]
    \centering
    \includegraphics[width=0.99\linewidth]{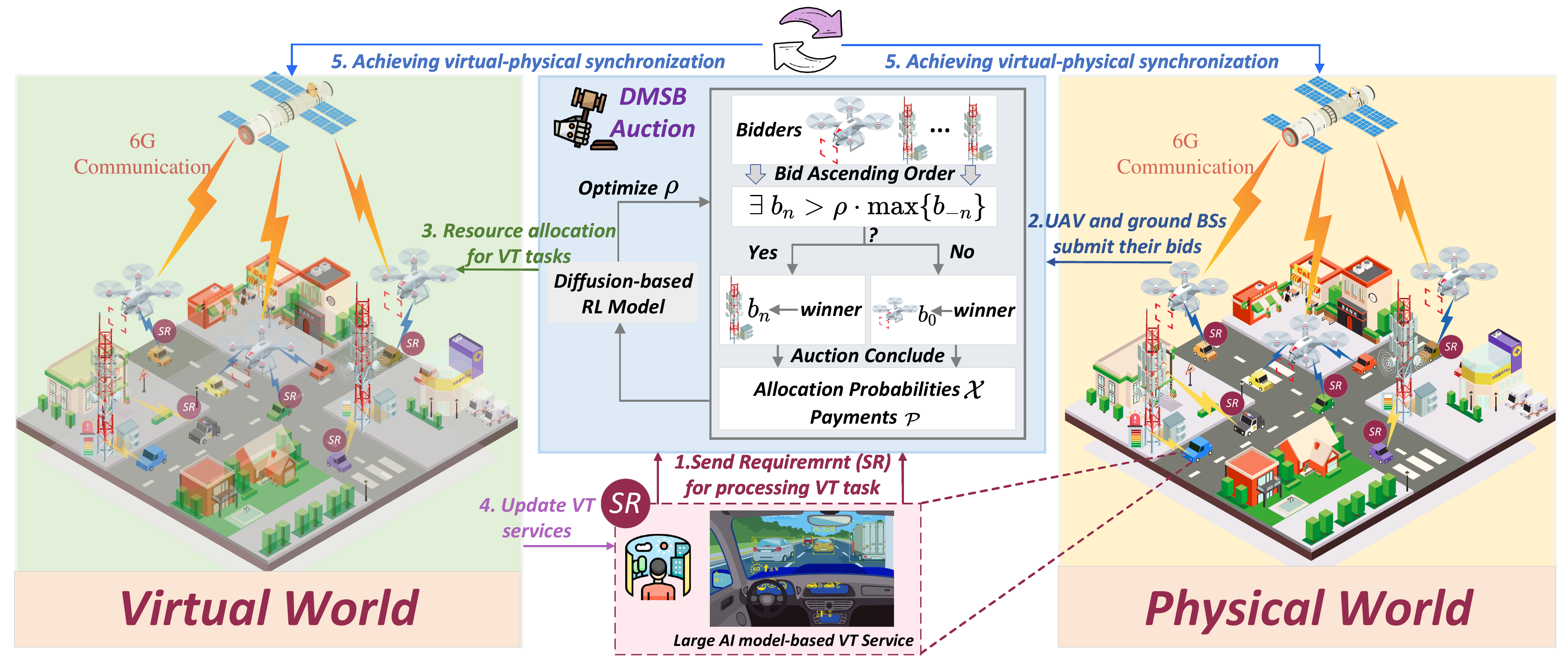}
    \caption{The system model of our proposed DMSB auction for resource allocation in 6G-enabled Vehicular Metaverses.}
    \label{framework_system}
\end{figure*}
\subsection{Network Model}
Vehicular users can access large AI model-based vehicular services via resource providers such as a UAV or ground BSs for data transmission. The uplink and downlink transmission rates determine the speed at which data can be offloaded to the resource providers and received back by the vehicular users. These rates depend on factors such as the channel power gains $g_{m,0}$ for UAV 0 and $g_{m,n}$ for ground BS $n=\{1,2,\cdots, N\}$, as well as the allocated bandwidth and transmit power. The uplink transmission rate for vehicular user $m$ to transmit input data of VT tasks to resource providers is defined as
 \begin{equation}
R_{m,n}^u=B_n^u\log_2\left(1+\frac{g_{m,n}p^u}{\sigma_n^2}\right),
 \end{equation}
where $p^u$ denotes the transmit power of vehicular user $m$, and $\sigma_n^2$ is the additive Gaussian white noise at ground BSs or a UAV~\cite{xu2023reconfiguring}. Similar to the uplink transmission rate, the downlink transmission rate at which the UAV or the ground BS transmits the result of the large AI model-based VT task to vehicular user $m$ is defined as
\begin{equation}
R_{m,n}^d=B_n^d\log_2\left(1+\frac{g_{m,n}p^d}{\sigma_m^2}\right),
\end{equation}
where $p^d$ is the transmit power of a UAV or ground BSs, and $\sigma_m^2$ denotes the additive Gaussian white noise at vehicular user $m$.
\subsection{Latency Model}
Before vehicular users can access immersive large AI model-based vehicular services, they need to offload a VT task with an input size of $D_{m}^{req}$. In wireless communication~\cite{10515202}, the latency generated by vehicular user $m$ offloading VT tasks inputs to resource providers (i.e., a UAV or ground BSs) at time slot $t$ depends on the uplink transmission $R_{m,n}^{u}$. Therefore, the uplink transmission latency for offloading the VT task input from vehicular user $m$ is calculated as 
\begin{equation}
    T_{m,n}^u=\frac{D_m^{req}}{R_{m,n}^u}.
\end{equation}
When a UAV or ground BSs receive the VT task $D_{m}^{req}$ from vehicular user $m$, it must process it, such as navigation or pedestrian detection, which often requires graphics rendering. Therefore, the processing latency of the VT task is determined by the available computation resources owned by a UAV or ground BSs, including GPU and CPU resources. Consequently, the processing latency is defined as
\begin{equation}
    T_{m,n}^p=\frac{D_m^{req}e^p}{f_n^g}+\frac{D_m^{req}}{f_n^c},
\end{equation}
where $e_p$ is the GPU resource required to process each bit of data, $f_n^g$ is the computational efficiency of GPUs and $f_n^c$ is the computational efficiency of CPUs. After processing the VT task $D_m^{req}$, the UAV or ground BS will send the completed task results back to vehicular user $m$. Specifically, when UAV or ground BS finished processing the VT task $D_m^{req}$, the UAV or ground BS generates a VT task of size $D_n^{task}$ according to the VT task $D_m^{req}$ input. Therefore, the downlink transmission latency that depends on $D_n^{task}$ and the downlink transmission rate $R_{m,n}^{d}$ is defined as
\begin{equation}
    T_{m,n}^d=\frac{D_n^{task}}{R_{m,n}^d}.
\end{equation}
The total latency of the VT task at time slot $t$ can be calculated as
\begin{equation}
    T_{m,n}^{t}=T_{m,n}^u+T_{m,n}^p+T_{m,n}^d.
\end{equation}
\subsection{Accuracy Model}
In the proposed system, accuracy is a key metric for evaluating the degree of pixel point matching in large AI model-based VT tasks. We consider a large AI model-based VT task (e.g., pedestrian detection task) requested by vehicular user $m$, which needs to be processed by a UAV or a ground BS. Each vehicular user $m$ captures images of its environment through a camera, consisting of $k$ pixel points denoted as $\Omega_m=\{w_m^1,w_m^2,\cdots,w_m^k\}$. These pixel points represent the visual information that need to be processed by a UAV or ground BSs. After the task is processed, the UAV or ground BS generates a set of processed pixel points $\Omega_n=\{c_n^1, c_n^2,\cdots, c_n^k\}$. To evaluate the accuracy of the VT task~\cite{singh2022ai}, we define pixel point matching count $D_m(\Omega_n,\Omega_m)$ as follows:
\begin{equation}
D_m(\Omega_n,\Omega_m)=\sum_{i=1}^k\mathbf{I}\{w_m^i=c_n^i\mid w_m^i\in\Omega_m,c_n^i\in\Omega_n\},
\end{equation}
where $\mathbf{I}\{\cdot\}$ is an indicator function that is equal to 1 if the condition is true, and 0 otherwise. Therefore, the accuracy of the VT task processing $R_{m}$ can be quantified by the proportion of correctly matched pixel points, which is defined as
\begin{equation}
    R_m=\frac{D_m(\Omega_n,\Omega_m)}k,
\end{equation}
where $R_{m}\in [0,1]$. The accuracy $R_{m}$ measures the degree to which the processed pixel points by the UAV or ground BS align with the original pixel points, indicating the effectiveness of the VT task processing.
\begin{figure*}[t]
    \centering
    \includegraphics[width=0.97\linewidth]{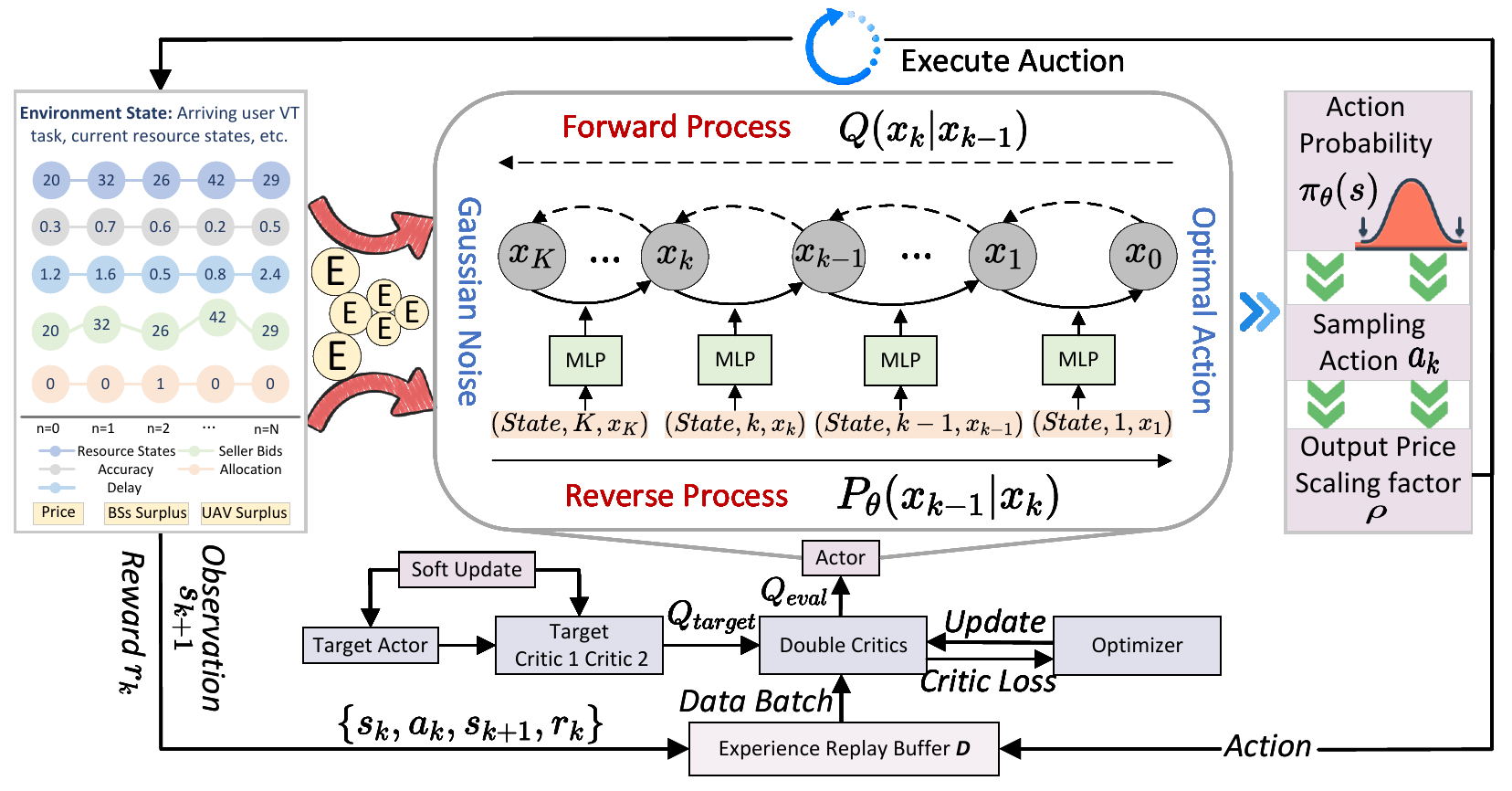}
    \caption{The architecture of the Diffusion-based RL algorithm for dynamic price scaling in the MSB auction.}
    \label{framework_algorithm}
\end{figure*}

\begin{algorithm}
\caption{Diffusion-based Modified Second-Bid (DMSB) Auction}
\label{alg:DMSB}
\begin{algorithmic}[1]
\STATE \textbf{Input:} Bids $b = \{b_0, b_1, \cdots, b_N\}$ from resource providers (UAV, ground BSs), diffusion step $K$, soft target update parameter $\varpi$.
\STATE \textbf{Output:} The price scaling factor $\rho$, allocation probabilities $\mathcal{X}$ and payments $\mathcal{P}$.
\STATE \textbf{Initialize:} Actor network parameters $\theta$, critic network parameters $\phi$, target actor network parameters $\theta^{\prime}$, target critic network parameters $\phi^{\prime}$, and experience replay buffer $\mathcal{D}$.
\FOR {episode $e = 1$ to $E$}
    \STATE Initialize the distribution $\boldsymbol{x}_T\sim\mathcal{N}(\mathbf{0},\mathbf{I})$.
    \FOR {iteration $k = 1$ to $\kappa$}
        \STATE Receive bids $b$ from resource providers and observe state $s_k = \{m^k, b^k\}$.
        \STATE Set $x_{K}$ as Gaussian noise and generate $x_{0}$ by denoising $x_{K}$ based on Eq.(\ref{reverse}).
        \STATE Execute the action $a_{k}$ based on $x_{0}$ and Eq.(\ref{probability}).
        \STATE Calculate the price scaling factor $\rho = 10^{a_k / |\mathcal{A}|}$.
        \STATE Execute $\rho$ in the auction environment, and observe the next state $s_{k+1}$ and current reward $r_k$.
        \STATE Store transition $(s_k, a_k, s_{k+1}, r_k)$ to the experience replay buffer $\mathcal{D}$.
        \STATE Sample a mini-batch of experiences $\hat{D}_k$ from the experience replay buffer $\mathcal{D}$.
        \STATE Update the actor network parameters $\theta$ using $\hat{D}_k$ by Eq.(\ref{theta}).
        \STATE Update the critic network parameters $\phi$ by minimizing the object function based on  Eq.(\ref{min}).
        \STATE Soft update the target actor network parameters $\theta^{\prime}$ by $\theta^{\prime}\leftarrow\varpi \theta+(1-\varpi )\theta^{\prime}$.
        \STATE Soft update the target critic network parameters $\phi^{\prime}$ by $\phi^{\prime}\leftarrow\varpi \phi+(1-\varpi )\phi^{\prime}$.
    \ENDFOR
\ENDFOR
\end{algorithmic}
\end{algorithm}

\section{Market Design And Problem Formulation}\label{Market Design}
%
\subsection{Market Design}
To incentivize resource providers to handle VT tasks, especially some large AI model-based VT tasks for vehicular users in 6G-enabled Vehicular Metaverses, we design a novel VT task market. In this market, there are three main entities: the auctioneer, vehicular users, and resource providers. 
\begin{itemize}
    \item Auctioneer: The auctioneer serves as the market's transaction platform, which is responsible for managing the distribution of VT tasks. The auctioneer receives VT task requests from vehicular users and makes these demands available to resource providers. Additionally, the auctioneer executes the auction process, ensuring that resources are efficiently allocated to the appropriate resource providers.
\end{itemize}
\begin{itemize}
    \item Vehicular Users: Vehicular users are the initiators of VT tasks, who submit VT task processing requirements to the auctioneer with the expectation that resource providers can complete VT tasks with the lowest latency and the highest accuracy. However, they do not participate in bidding in the market because their primary role is to send VT task requirements, while the auction mechanism is designed to optimize resource allocation among competing resource providers based on these requirements.
\end{itemize}
\begin{itemize}
    \item Resource Providers: Resource providers, including a UAV and serval ground BSs, act as sellers (i.e., bidders) within the market, competing to offer VT task processing for vehicular users to obtain revenue. Additionally, we consider that resource providers are risk-neutral bidders, with their surpluses being positively correlated according to the revelation principle~\cite{arnosti2016adverse}.
\end{itemize}

In 6G-enabled Vehicular Metaverses, vehicular users can access VT task processing by requesting resources from the resource providers. In the VT task market, the valuation $v_{n}$ for resource provider $n$ for the opportunity to process a VT task is the product of common value $c_{n}$ and matching value $m_{n}$. The common value $c_{n}$ reflects the average resource consumption of a UAV or a ground BS when processing VT tasks, based on the expected latency over the past $T$ time periods, which is defined as
\begin{equation}
    c_n = \mathbb{E}_{t\in T} \left[ \frac{\omega_1}{T^t_{m,n}} \right],
\end{equation}
where $\omega_1$ is a latency-related value factor. To capture the variability in accuracy across different VT tasks over the past $T$ time periods, we model the matching value $m_{n}$ as an expected value based on the accuracy $R_m$, which is calculated as
\begin{equation}
    m_n = \mathbb{E}_{t\in T} \left[ \frac{\omega_2}{(1-R_{m})^\beta} \right],
\end{equation}
where $\beta$ is a parameter controlling the sensitivity of the matching value to changes in accuracy and $\omega_2$ is a constant that scales the impact of accuracy on the matching value. Therefore, the valuation $v_{n}$ of resource provider $n$ can be calculated as
\begin{equation}
    v_{n}=\mathbb{E}_{t\in T} \left[ \frac{\omega_1}{T^t_{m,n}} \right] \cdot \mathbb{E}_{t\in T} \left[ \frac{\omega_2}{(1-R_{m})^\beta} \right].
\end{equation}
When the resource provider $n$ values the VT task,  the common value $c_n$ and matching value $m_n$ are considered independent~\cite{arnosti2016adverse}. This means that the resource consumption, represented by the common value, is independent of the quality of task processing, which is represented by the matching value.

\subsection{Problem Formulation}
In the proposed VT task market, we design a mechanism $\mathcal{M}=(\mathcal{X},\mathcal{P})$, which maps the private valuation $v_n$ of resource providers into the VT task allocation probabilities $\mathcal{X}$ and payments $\mathcal{P}$. In this mechanism, resource providers submit their bids $b=\{b_0,b_1,\cdots,b_N\}$ based on their valuation $v_n$, representing their willingness to process VT tasks. Then, the auctioneer determines the allocation probabilities $\mathcal{X}=\{\mathcal{X}^{UAV}(b_0),\mathcal{X}^{BS}(b_1),\ldots,\mathcal{X}^{BS}(b_N)\}$ and the payments $\mathcal{P}=\{\mathcal{P}^{UAV}(b_0),\mathcal{P}^{BS}(b_1),\ldots,\mathcal{P}^{BS}(b_N)\}$ based on these bids, where $\mathcal{X}^{UAV}(b_0)$ indicates the probability that the VT task is allocated to a UAV, and $\mathcal{X}^{BS}(b_n)$ denotes the probability of allocating the VT task to ground BS $n$. The expected surplus depends on the realization of the valuations and the allocation, for UAV is $S^{UAV}=\mathbb{E}[v_0\cdot \mathcal{X}^{UAV}(b_0)]$, for ground BSs are calculated as $S^{BSs}=\mathbb{E}[\sum_{n=1}^Nv_n\cdot \mathcal{X}^{BS}(b_n)]$. The objective of the auctioneer is to maximize the total expected surplus from the VT task allocation, which can be formulated as
\begin{subequations}
\begin{align} 
\max_{\mathcal{M}} \quad &\zeta \cdot S^{UAV} + S^{BSs} \\
\text{s.t.} \quad & b_n = v_n, \label{13} \\ 
& \sum_{n=1}^{N} \mathcal{X}^{BS}(b_n) + \mathcal{X}^{UAV}(b_0) = 1, \label{eq14} \\ 
& \mathcal{X}^{BS}(b_n) \in \{0, 1\}, \quad n \neq 0, \label{15} \\ 
& \mathcal{X}^{UAV}(b_0) \in \{0, 1\}. \label{16}
\end{align}
\end{subequations}
Constraint~(\ref{13}) indicates that resource providers submit bids based on their valuation. Constraint~(\ref{eq14}) ensures that a VT task is allocated to only one resource provider. Constraints~(\ref{15}) and~(\ref{16}) indicate whether a VT task is fully allocated to the ground BS or a UAV.

\section{Diffusion-based Modified Second-bid Auction}\label{Diffusion-based}
To address the challenges of VT task resource allocation in 6G-enable Vehicular Metaverses, we propose the DMSB auction mechanism, which optimizes the bidding process for resource providers to alleviate information asymmetry among them. We first introduce the MSB auction~\cite{arnosti2016adverse}, in which the winning bidders and payment depend on a price scaling factor $\rho$. Then, we introduce the proposed Diffusion-based RL algorithm to dynamically adjust the price scaling factor $\rho$~\cite{du2024diffusion}, maximizing the total surplus in resource providers. Finally, we prove that the proposed DMSB auctions are fully strategy-proof and adverse-selection-free.
\subsection{Modified Second-Bid Auction}
In the MSB auction, the auctioneer collects bids $b=\{b_0,b_1,\cdots,b_N\}$ from resource providers, including the UAV acting as brand bidder and ground BSs acting as performance bidders. The brand bidder refers to the UAV, which like traditional brand advertisers~\cite{arnosti2016adverse}, focuses on broad objectives such as improving coverage or network reliability, but cannot accurately assess the value of individual VT tasks. Conversely, performance bidders are ground BSs, can evaluate each VT task based on key performance metrics~\cite{arnosti2016adverse}, such as low latency and high accuracy, and bid accordingly.
The MSB auction allocates VT tasks by comparing the submitted bids and determining the winner based on a price scaling factor $\rho$. Specifically, the highest ground BS wins the VT task if it exceeds the second-highest bid by at least $\rho$, and the payment is determined by scaling the second-highest bid with $\rho$. If no ground BS satisfies this condition, the VT task is allocated to the UAV, and the contracted payment is determined by maximizing the expected profit on the highest possible value of $b_{0}$~\cite{arnosti2016adverse}, calculated as
\begin{equation}
    b_0=\max_b\mathbb{E}\begin{bmatrix}(v_0-v_{max})1(v_{max}\leq b_0)\end{bmatrix},
\end{equation}
where $v_0$ is the valuation of the VT task measured by the UAV, and $v_{max}$ is the highest valuation. Therefore, the allocation rule and the pricing rule can be expressed as follows.
\begin{itemize}
    \item Allocation rule: For performance bidders (i.e., ground BSs), the allocation probability $\mathcal{X}^{BS}(b_n)$ is calculated as
    \begin{equation}
        \mathcal{X}^{BS}(b_n)=
        \begin{cases}1, & \text{ if } b_n>\rho\cdot\max\{b_{-n}\}  \\
        0,&\text{ otherwise.}
        \end{cases}
    \end{equation}
\end{itemize}
In this allocation rule, the performance bidder $n$ wins the MSB auction and is allocated the VT task if its bid $b_n$ exceeds the maximum bid $b_{-n}$ of all other providers by a price scaling factor $\rho$ and $\rho \ge 1 $. Otherwise, the brand bidder (i.e., UAV) wins the MSB auction if no performance bidder wins. Therefore, the allocation probability of the UAV can be calculated as $\mathcal{X}^{UAV}(b_0)\le 1-\sum_{n=1}^{N} \mathcal{X}^{BS}(b_n)$.
\begin{itemize}
    \item Pricing rule: The payment $\mathcal{P}=\{\mathcal{P}^{UAV}(b_0),\mathcal{P}^{BS}(b_1),\ldots,\mathcal{P}^{BS}(b_N)\}$ for the resource providers is based on the second-highest bid scaled by $\rho$. Specifically, the payment rule for performance bidder can be calculated as
    \begin{equation}
                \mathcal{P}^{BS}(b_n) = \mathcal{X}^{BS}(b_n) \cdot \rho \cdot\max\{b_{-n}\} .
    \end{equation}
    For band bidders, if the UAV wins the auction, its payment is equal to its contracted payment $b_0$, i.e., $\mathcal{P}^{UAV}(b_0)=\mathcal{X}^{UAV}(b_0) \cdot b_{0}$.
\end{itemize}
\begin{figure}
    \centering
\includegraphics[width=0.6\linewidth]{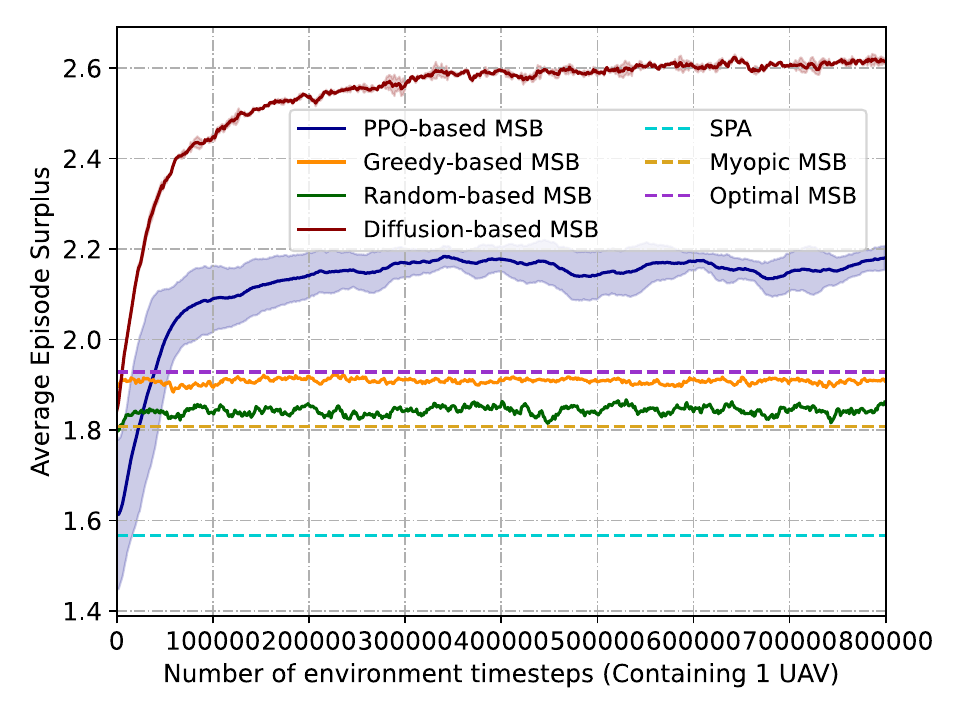}
    \caption{Convergence of the DMSB auction scheme compared with PPO, Greedy, Random, and theoretical auction methods.}
    \label{fig_cove_all}
\end{figure}

\subsection{Diffusion-based Pricing Scaling Factor}
In this section, we introduce the Diffusion-based RL approach to dynamically adjust the price scaling factor $\rho$, as shown in Fig.~\ref{framework_algorithm}, which influences the allocation and payment rules in DMSB auctions. By leveraging the Diffusion-based RL model, the auctioneer learns and adapts $\rho$ based on historical bidding data, real-time system load, and market dynamics, to maximize the total surplus across resource providers.

The Diffusion-based RL determines MSB's pricing scaling factor by formulating it as a Markov Decision Process (MDP), where the auction environment consists of states, actions, and rewards. During each auction round, the auctioneer receives bids from resource providers (i.e., a UAV and ground BSs). The state space $s_k$ at timestep $k$ consists of the current market conditions $m^k$ and the submitted bids $b^k$, i.e., $s_{k}=\left \{ m^k,b^k \right \} $. At each timestep, the auctioneer's action is to adjust the price scaling factor $\rho$, defined as $a_k\in\mathcal{A}, \rho=10^{a_k/|\mathcal{A}|}$, where $|\mathcal{A}|$ represents the size of the action space. The reward for each step corresponds to the total surplus generated in the MSB auction, which can be expressed as 
$r_k=\mathbb{E}\left[\zeta \cdot S^{UAV} + S^{BSs}\right]$, where $S^{UAV}$ and $S^{BSs}$
are the surplus from UAV and ground BSs, respectively.

The auctioneer employs a diffusion model to learn a likelihood distribution of potential actions, which influences the value of $\rho$. The diffusion model consists of two phases: a forward process and a reverse process. In the forward diffusion process, Gaussian noise is added progressively to the original input state over multiple timesteps, enabling exploration of a wide range of possible actions by gradually transforming the state representation. At each timestep $k$, the input data $x_k$ is generated from the previous state $x_{k-1}$ by applying Gaussian noise~\cite{du2024diffusion}:
\begin{equation}
    Q(x_k|x_{k-1})=\mathcal{N}(x_k;\sqrt{1-\eta_k}x_{k-1},\eta_kI),
\end{equation}
where $\eta_k$ controls the variance of the noise, and $I$ represents the identity matrix. 
Once the forward process reaches the terminal state, the reverse diffusion process begins. In this stage, the model denoises the data, gradually reconstructing the optimal action (i.e., the optimal $\rho$) from the noisy state. The reverse process is defined as:
\begin{equation}\label{reverse}
    P_\theta(x_{k-1}|x_k)=\mathcal{N}(x_{k-1};\Lambda_\theta(x_k,k,s),\sigma_k^2I),
\end{equation}
where $\Lambda_\theta$ is the learned mean function conditioned on the noisy data $x_k$, the current state $s$, and the timestep $k$. 
After completing the reverse process, the final output, $x_0$, is transformed into a probability distribution over the action space using the softmax function:
\begin{equation}\label{probability}
    \pi_\theta(s)=\left\{\frac{e^{x_0^i}}{\sum_{k=1}^{\mathcal{A}}e^{x_0^k}},\forall i\in\mathcal{A}\right\},
\end{equation}
where each value in $\pi_{\theta}(s)$ represents the likelihood of choosing a particular price scaling factor $\rho$.

During the interaction between the learning agent (i.e., the auctioneer) and the MSB environment, it collects transitions $(s_{k}, a_{k}, s_{k+1}, r_{k})$. These transitions are stored in an experience replay buffer $D$, enabling the agent to revisit past experiences and improve its policy through iterative learning. 
To evaluate the quality of the selected actions and ensure that the chosen price scaling factor $\rho$ maximizes the total surplus, we incorporate two critic networks, $Q_{\phi_1}(s, a)$ and $Q_{\phi_2}(s, a)$. These critic networks estimate the value of taking action $a$ (i.e., adjusting $\rho$) in state $s$. 
The parameters $\phi_1$ and $\phi_2$ are trained by minimizing the Bellman error, which is expressed as:
\begin{equation}\label{min}
\min_{\phi_i}\mathbb{E}_{(s_k,a_k,s_{k+1},r_k)\sim \hat{D}_{k}}\left[\left(Q_{\phi_i}(s,a)-y\right)^2\right],
\end{equation}
where $i \in \{1, 2\}$, and $\hat{D}_{k}$ is a randomly sampled mini-batch of transitions retrieved from the experience replay buffer $D$ during the $k$-th training iteration, while the target value $y$ is calculated as:
\begin{equation}
    y=r_k+\gamma\min_{i=1,2}Q_{\phi_i^{\prime}}(s_{k+1},\pi_{\theta^{\prime}}(s_{k+1})),
\end{equation}
with $\phi_i^{\prime}$ and $\theta^{\prime}$ are the parameters of the target critic networks and target actor-network, respectively.

The actor-network parameters, represented by $\theta$, are optimized by using gradient descent to minimize the expected loss while promoting exploration through entropy regularization:
\begin{equation}\label{theta}
    \theta\leftarrow\theta-\hat{\eta}\nabla_\theta\left(\mathbb{E}_{s_{k}\sim \hat{D}_{k},a\sim\pi_\theta}\left[-\min_{i=1,2}Q_{\phi_i}(s,a)\right]-\bar{\beta} H(\pi_\theta(s))\right)
\end{equation}
where $H(\pi_{\theta}(s))$ is the entropy of the policy, $\hat{\eta}$ is the learning rate, and $\bar{\beta}$ is the temperature parameter that balances exploration and exploitation. More details about the DMSB auction are shown in Algorithm~\ref{alg:DMSB}.

The computational complexity for Algorithm~\ref{alg:DMSB} is $O(E\kappa [V+K|\theta |+(\varrho +1)(|\theta |+|\phi |)])$. This complexity consists of three major parts:
\begin{itemize}
    \item Environment interaction cost: $O(E\kappa V)$. During each episode $E$ and each iteration $\kappa $, the algorithm needs to observe the current state, execute actions, and observe the next state and corresponding rewards, with $V$ denoting the complexity of the auction environment interaction.
\end{itemize}
\begin{itemize}
    \item Reverse diffusion process overhead: $O(E\kappa K|\theta |)$~\cite{du2024diffusion}. This overhead arises from the reverse diffusion process, which consists of $K$ denoising steps of neural network inference using the actor networks, where the $|\theta |$ represents the number of parameters in the actor networks.
\end{itemize}
\begin{itemize}
    \item Parameter update cost:$O(E\kappa (\varrho +1)(|\theta |+|\phi |)])$~\cite{du2024diffusion}. This term includes the policy improvement cost $O(E\kappa \varrho|\theta |))$, the critic network parameters improvement cost $O(E\kappa \varrho|\phi |))$, and the target network update cost $O(E\kappa (|\theta |+|\phi |))$. Here, $\varrho$ represents the mini-batch size used for updating the networks, and $|\phi |$ represents the number of parameters in the critic networks.
    \end{itemize}
\begin{figure}[!t]
    \centering

    \subfloat[\label{fig:3_1_convergence}]{
        \includegraphics[width=0.4\textwidth]{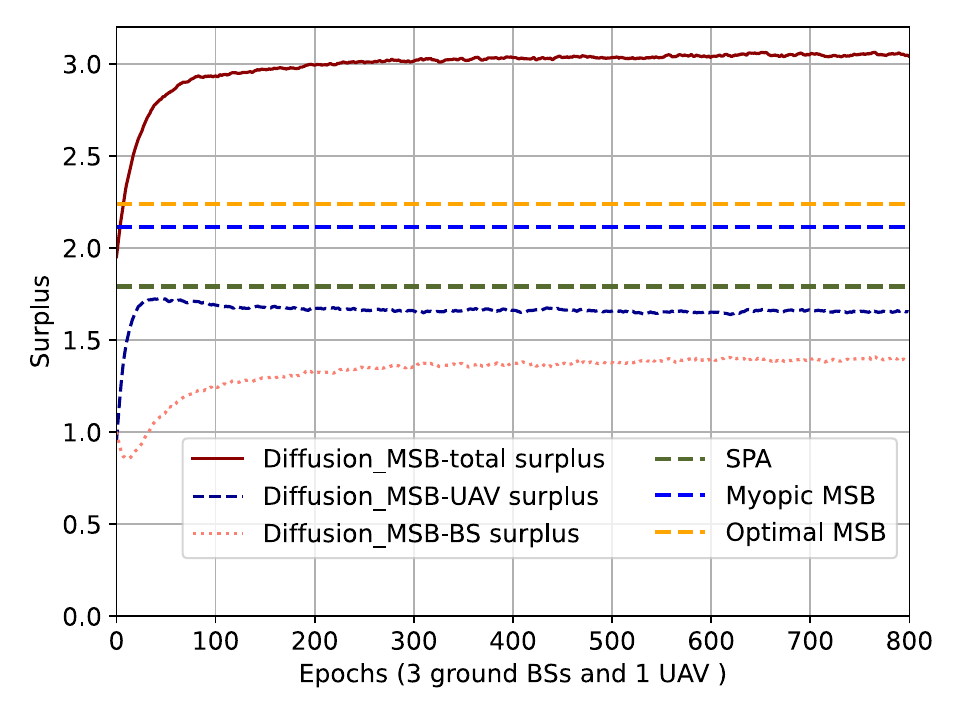}
    }
    \hspace{0.02\textwidth} 
    \subfloat[\label{fig:5_1_convergence}]{
        \includegraphics[width=0.4\textwidth]{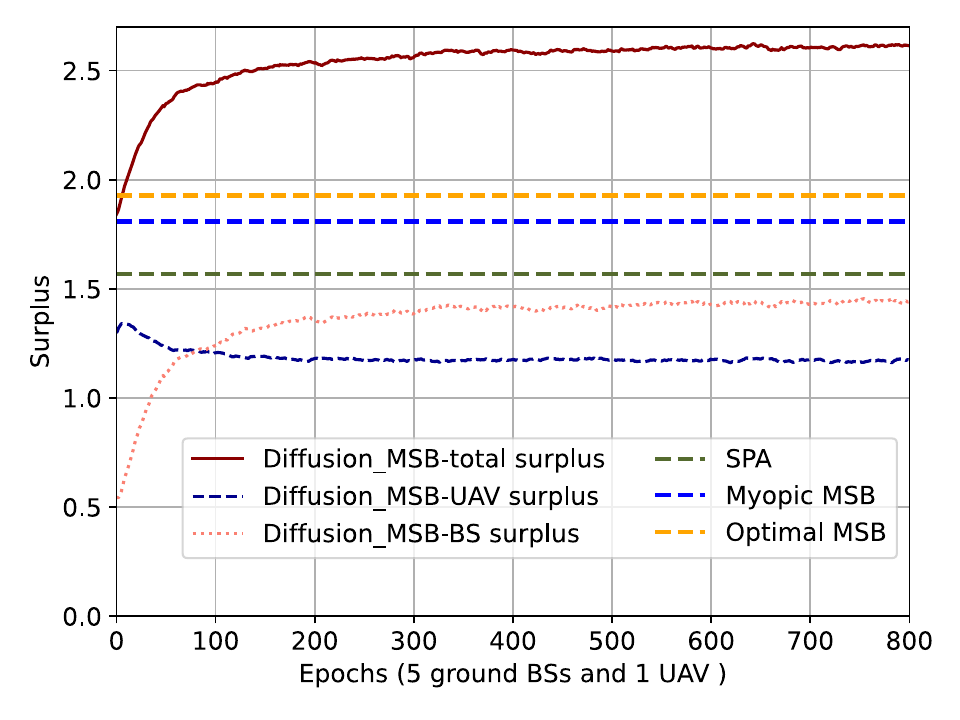}
    }

    \vspace{0.02cm} 

    \subfloat[\label{fig:7_1_convergence}]{
        \includegraphics[width=0.4\textwidth]{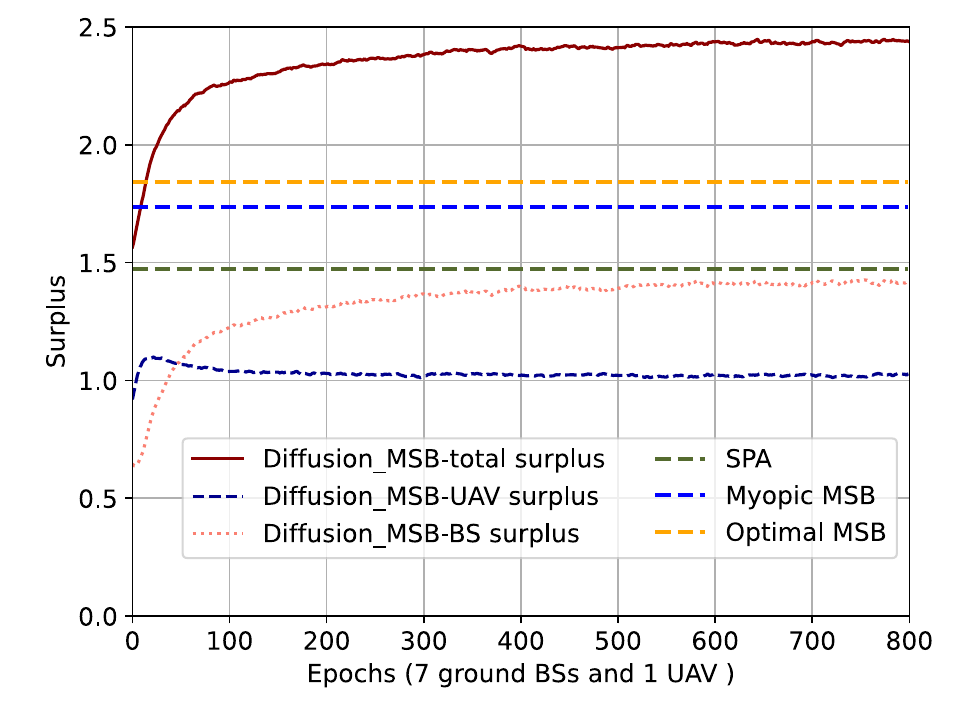}
    }
    \hspace{0.02\textwidth} 
    \subfloat[\label{fig:9_1_convergence}]{
        \includegraphics[width=0.4\textwidth]{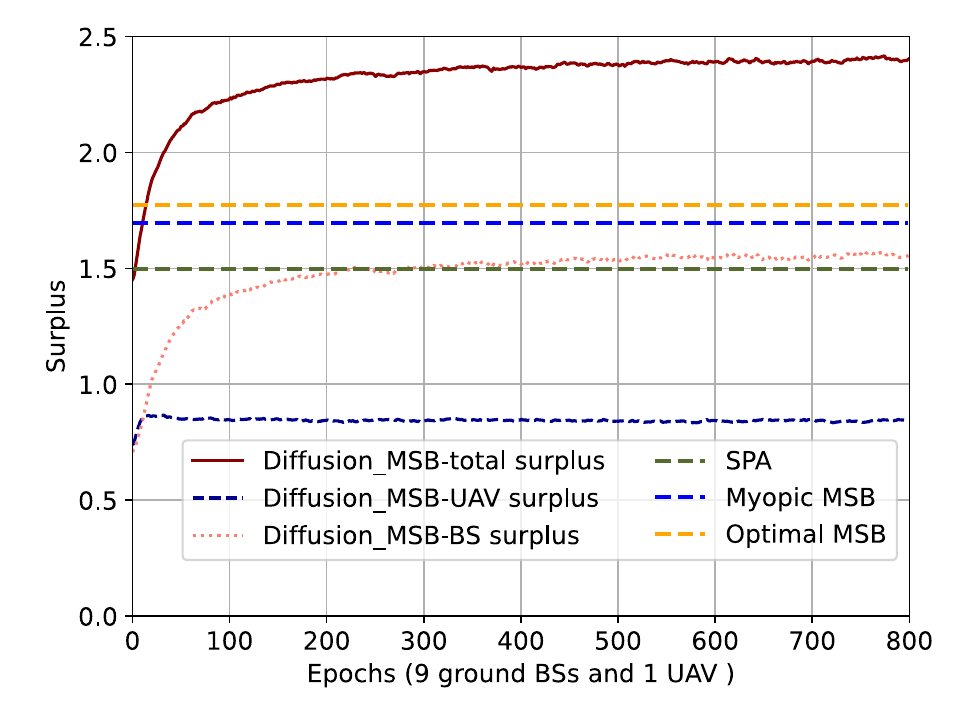}
    }

    \caption{Convergence of the proposed DMSB auction under surplus compared with the different number of resource providers. }
    \label{fig:overall-convergence-plots}
\end{figure}
\subsection{Property Analysis}
In an auction, strategy-proofness guarantees that bidders cannot increase their utility by misreporting their true bids. In addition, the absence of adverse selection implies that factors such as market externalities or information asymmetries do not affect bidder results.
Thus, it is necessary to prove that the proposed DMSB auction is entirely strategy-proof~\cite{milgrom2021auction} and free from adverse selection~\cite{arnosti2016adverse}, as presented in the following theorem.

\begin{theorem}\label{T1}
The DMSB auction, utilizing a dynamic price scaling policy governed by parameters $\rho$ and optimized through the Diffusion-based RL algorithm, maintains anonymity, fully strategy-proof, and is free from adverse selection.
\end{theorem}

\begin{proof}
    To demonstrate that the proposed DMSB auction is anonymous, fully strategy-proof, and free from adverse selection, we define a critical payment function, denoted as $\psi$. Let $\psi(b_{-n};\rho)$ represent the critical payment required for ground BS bidder $n$ to win, given the competing bids $b_{-n}$ and the price scaling factor $\rho$. Ground BS $n$ is declared the winner when its bid surpasses the critical payment $\psi(b_{-n};\rho)=\rho\max\{b_{-n}\}$, where $\rho=10^{a_k/|\mathcal{A}|}$. Upon winning, ground BS $n$ must pay $\psi(b_{-n};\rho)$. Given that $\rho=10^{a_k/|\mathcal{A}|} \ge 1$, only the highest bidder wins the auction, satisfying $\psi(b_{-n};\rho) \ge \max\{b_{-n}\}$. Moreover, the critical payment function of the DMSB auction satisfies the following conditions,
    \begin{equation}
\psi(\max\{b_{-n}\};\rho)=\rho\cdot\max\{\max\{b_{-n}\}\}=\rho\cdot\max\{b_{-n}\}=\psi(b_{-n};\rho).
    \end{equation}
   To ensure strategy-proofness, consider a scenario with two bidders in the market, where one bid exceeds $\psi(b_{-n};\rho)$ and the other is exactly equal to $\max\{b_{-n}\}$. In this case, if $\psi(\max\{b_{-n}\};\rho)\ne\psi(b_{-n};\rho)$, the auction could fail to satisfy the false-name-proof condition. 
    Specifically, when $\psi(b_{-n};\rho)<\psi(\max\{b_{-n}\};\rho)$, the first bidder may lower its bid without altering the remaining bids in the set $b_{-n}$, leading the DMSB auction fails to be winner-false-name-proof because lowering the bid would not change the bidder's winning status.
    Otherwise, if $\psi(b_{-n};\rho)>\psi(\max\{b_{-n}\};\rho)$, the DMSB auction does not prevent a losing bidder from submitting a higher bid to surpass that of the winner without altering the other bids in $b_{-n}$.
    In this case, the losing bidder could raise their bid relative to the winning bidder's bid and potentially change the auction result, violating loser-false-name-proofness. 
    Furthermore, the critical payment function $\psi(b_{-n};\rho)$ exhibits homogeneity of degree one, indicating that the auction is free from adverse selection. If $\psi(b_{-n};\rho)$ was not homogeneous of degree one, bidders could exploit their private information $\hat{\vartheta} \in \left \{ 1,\vartheta  \right \}$ to strategically modifying their bid, manipulating the results of DMSB auctions~\cite{xu2024cached}.
    Specifically, if $\psi(\mathbf{m}_{-n};\rho)<\psi(\vartheta\mathbf{m}_{-n};\rho)/\vartheta$, where $\vartheta\in\mathbb{R}_+,n\geq2$, $b_{-n}\in\mathbb{R}_+^{n-1}$ and $\hat{\vartheta}=1$, then $\mathcal{X}^{BS}(\hat{\vartheta}\mathbf{m})=\mathcal{X}^{BS}(\mathbf{m})=1_{\{m_{n}>\psi(m_{-n};\rho)\}}=1$, meaning $\mathcal{X}^{UAV}(\hat{\vartheta}\mathbf{m})=0$. When $\hat{\vartheta} \ne 1$, the bidders could modify their bids from $\hat{\vartheta}=\vartheta$, thereby altering their chances of winning the auction, i.e., $\mathcal{X}^{BS}(\hat{\vartheta}\mathbf{m})=\mathcal{X}^{BS}(\vartheta\mathbf{m})=1_{\{\vartheta m_{n}>\psi(\vartheta m_{-n};\rho)\}}= 0$. In this case, ground BS $n$ cannot be the winner despite submitting the highest bid, which results in the UAV bidder 0 winning the
    auction, i.e., $\mathcal{X}^{UAV}(\hat{\vartheta}\mathbf{m})=1$.

    From the above analysis, it is demonstrated that the proposed DMSB auction ensures anonymity, is entirely strategy-proof, and is free from adverse selection.
\end{proof}
\begin{figure}[!t]
    \centering

    \subfloat[Total Surplus versus Bandwidths\label{fig:revenue_bandwidth}]{
        \includegraphics[width=0.3\textwidth]{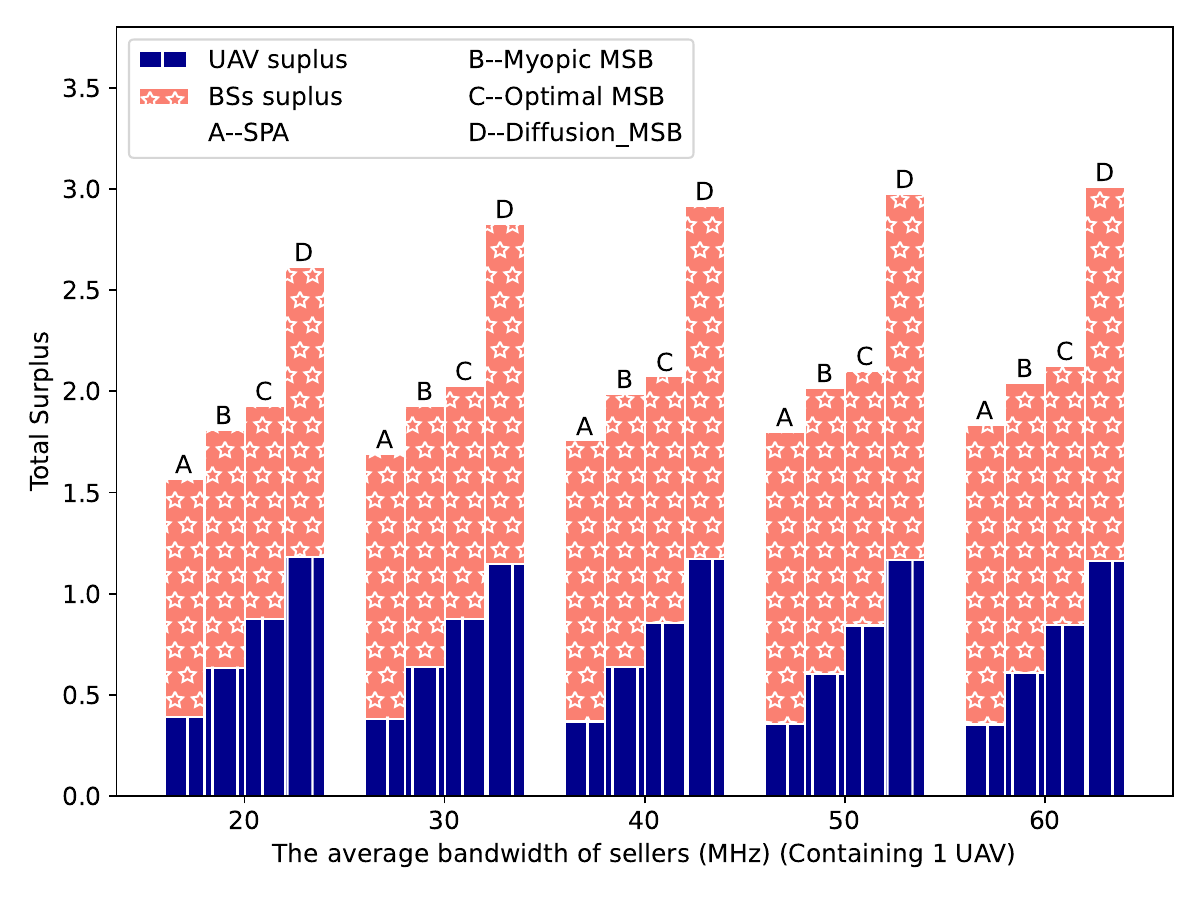}
    }
    \hspace{0.02\textwidth} 
    \subfloat[Total Surplus versus Seller Amount\label{fig:revenue_seller}]{
        \includegraphics[width=0.3\textwidth]{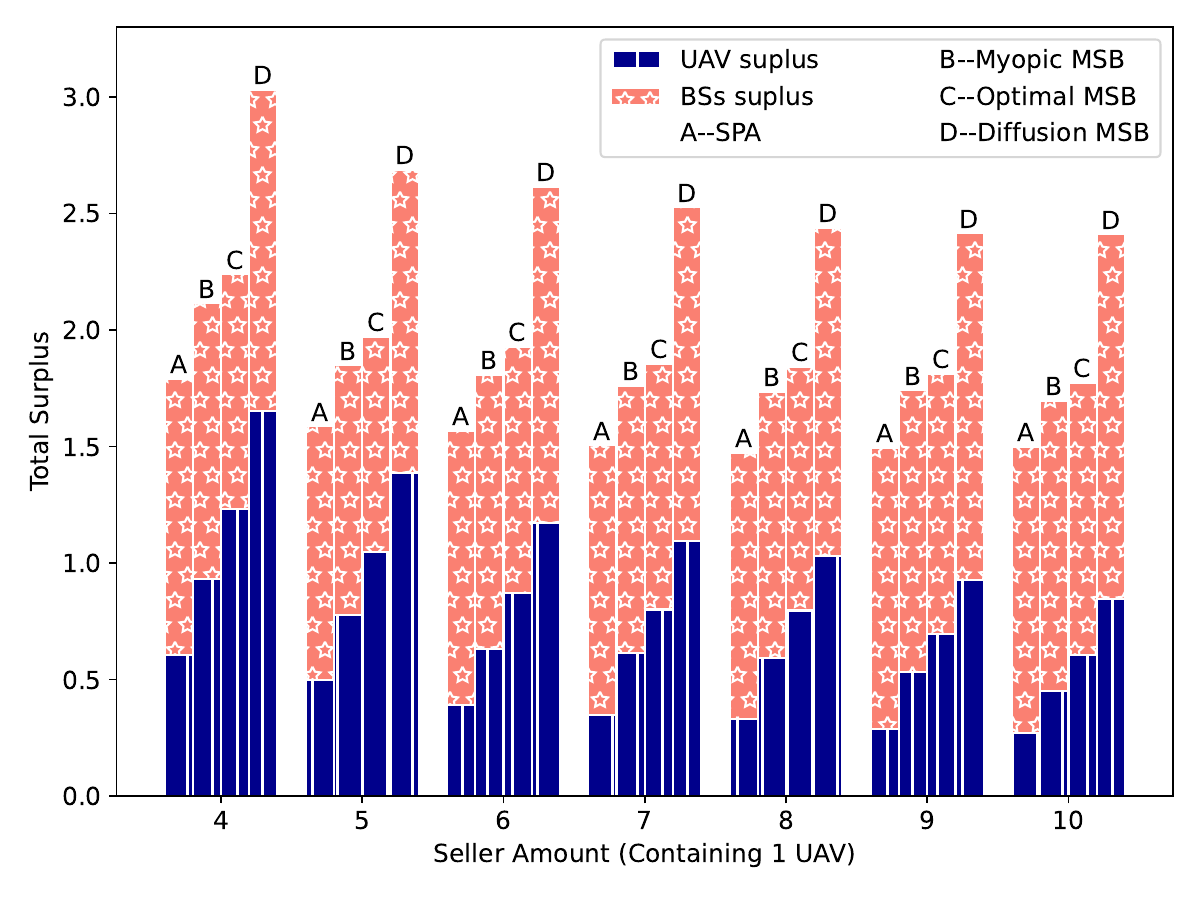}
    }
    \hspace{0.02\textwidth} 
    \subfloat[Total Surplus versus VT Sizes\label{fig:revenue_vt_size}]{
        \includegraphics[width=0.3\textwidth]{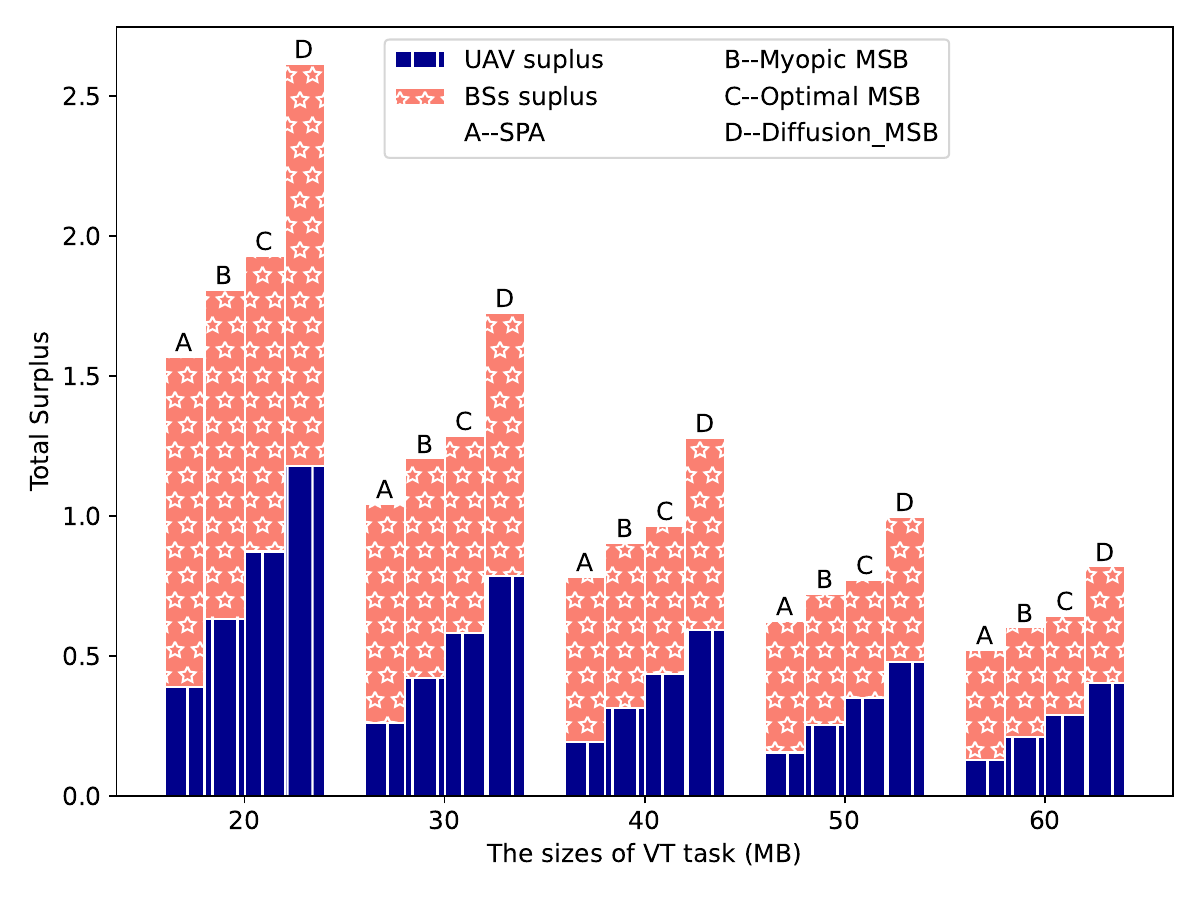}
    }

    \caption{Comparison of total surplus under different DMSB auction environment settings.}
    \label{fig:overall-revenue-plots}
\end{figure}

\section{Experimental Results}\label{experiment}
In this section, we first present the detailed experimental setup for the proposed DMSB auction in 6G-enabled Vehicular Metaverses. Next, we analyze the convergence of the proposed DMSB auction compared with Proximal Policy Optimization (PPO)-based MSB auction, Greedy-based MSB auction, and Random-based MSB auction methods. Finally, we conduct extensive experiments under different market settings to validate the effectiveness and performance of the proposed DMSB auction.
\subsection{Parameter Settings}
In the DMSB auction market, multiple resource providers, including ground BSs and a UAV, compete to process VT tasks containing large AI model-based VT tasks by submitting bids. In our simulation, the number of resource providers is set to 3, 4, 5, 6, 7, 8, and 9 ground BSs, with 1 UAV. The number of VT task requests is fixed at 1, and the size of each VT task is set between 20 and 40 MB to simulate different computing demands. Specifically, when a UAV or ground BS completes the processing of a VT task $D_m^{req}$, it generates another VT task of size $D_n^{task}$ based on the initial task input, assuming the sizes of $D_m^{req}$ and $D_n^{task}$ are identical.
The uplink bandwidth $B_{n}^{u}$ and downlink bandwidth $B_{n}^{d}$ for each ground BS and UAV vary between 20 MHz and 60 MHz. The computational efficiency of GPUs $f_{n}^{g}$ and CPUs $f_{n}^{c}$ are randomly allocated, with a maximum of 2 units per auction round. The GPU resource required to process each bit of data $e_{p}$ is set to 0.5. Additionally, the transmission power $p^d$ for a UAV and ground BSs, as well as $p^{u}$ for vehicular users, is randomly selected between 1 and 10 units. The sensitivity parameter $\beta$, which controls the responsiveness of the matching value $m_{n}$ to the changes in accuracy, is set to 2.

The auctioneer uses a Diffusion-based RL algorithm to dynamically optimize both the allocation and payment strategies by adjusting the price scaling factor $\rho$. We compare the convergence of the Diffusion-based RL algorithm with other baseline methods, such as PPO-based MSB auction, Greedy-based MSB auction, and Random-based MSB auction methods. 
In addition, we also compare the performance of the proposed DMSB auction with other auction baselines, including Second-Price Auction (SPA), Myopic MSB, and Optimal MSB~\cite{arnosti2016adverse, xu2024cached}.
Specifically, the price scaling factor $\rho$ in the Myopic MSB auction is set as $\rho=\max(1,u_0/u_{(2)})$, where $u_0$ is the highest bid, and $u_{(2)}$ is the second-highest bid, using information from the current round. In contrast, the price scaling factor $\rho$ in the Optimal MSB auction is set as $\rho=\max(1,\mathbb{E}_{[u_0]}/\mathbb{E}_{[u_{(2)}]})$~\cite{arnosti2016adverse}, based on historical statistical information.
\begin{figure}[!t]
    \centering

    \subfloat[\centering Total surplus versus Different bandwidths\label{bandwidth}]{
        \includegraphics[width=0.4\textwidth]{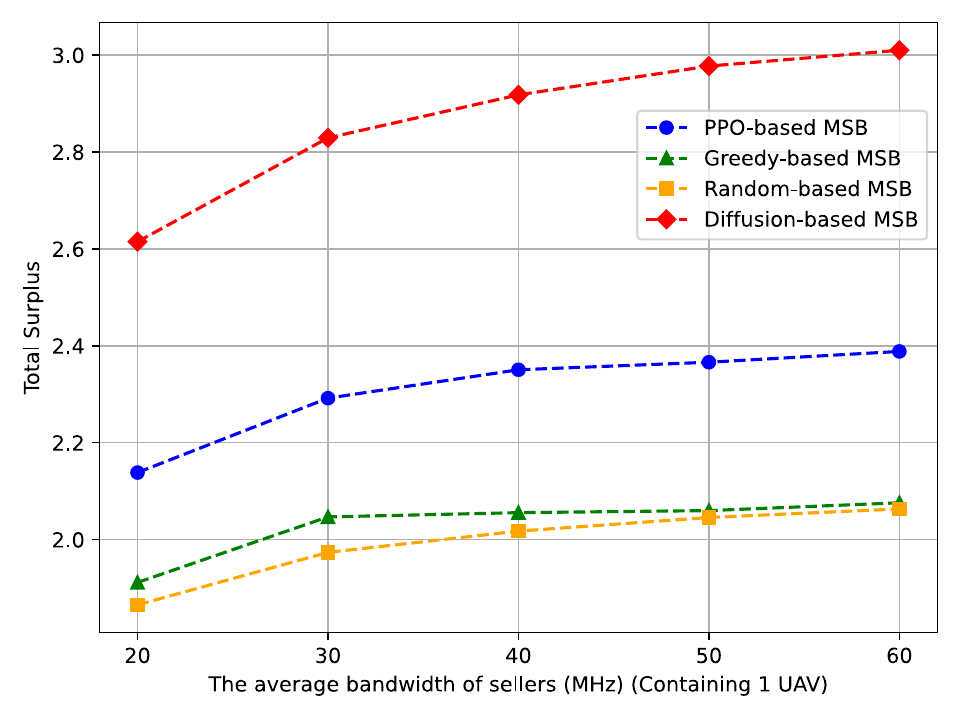}
        }
    \hspace{0.02\textwidth} 
    \subfloat[\centering Total surplus versus Different seller amounts\label{selleramount}]{
        \includegraphics[width=0.4\textwidth]{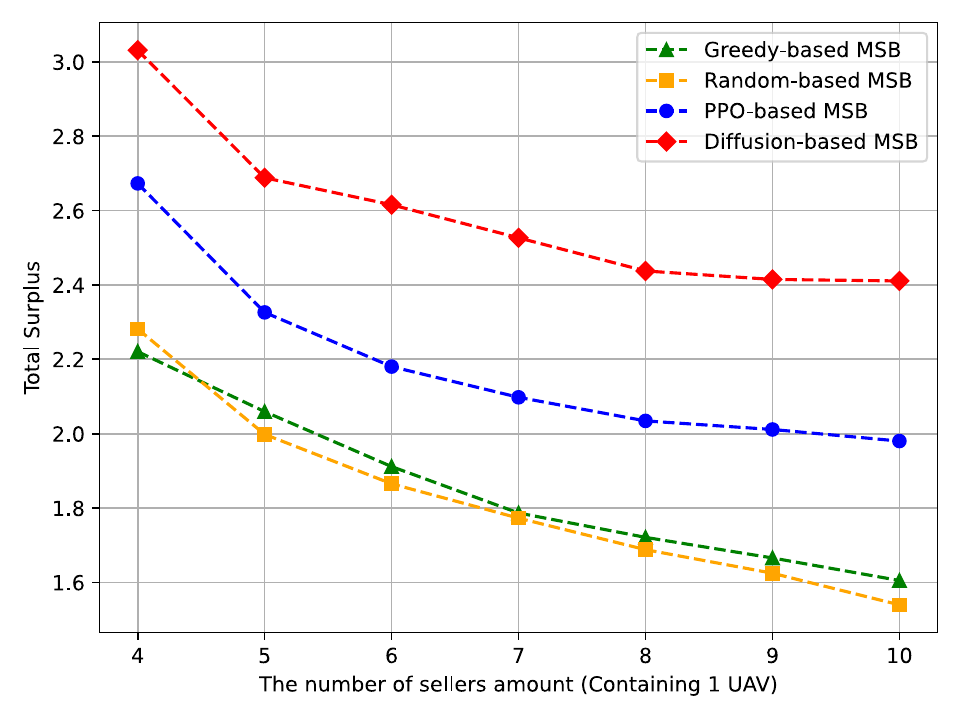}
    }
    \vspace{0.05cm} 
    \subfloat[\centering Total surplus versus Different sizes of VT task\label{VTsize}]{
        \includegraphics[width=0.4\textwidth]{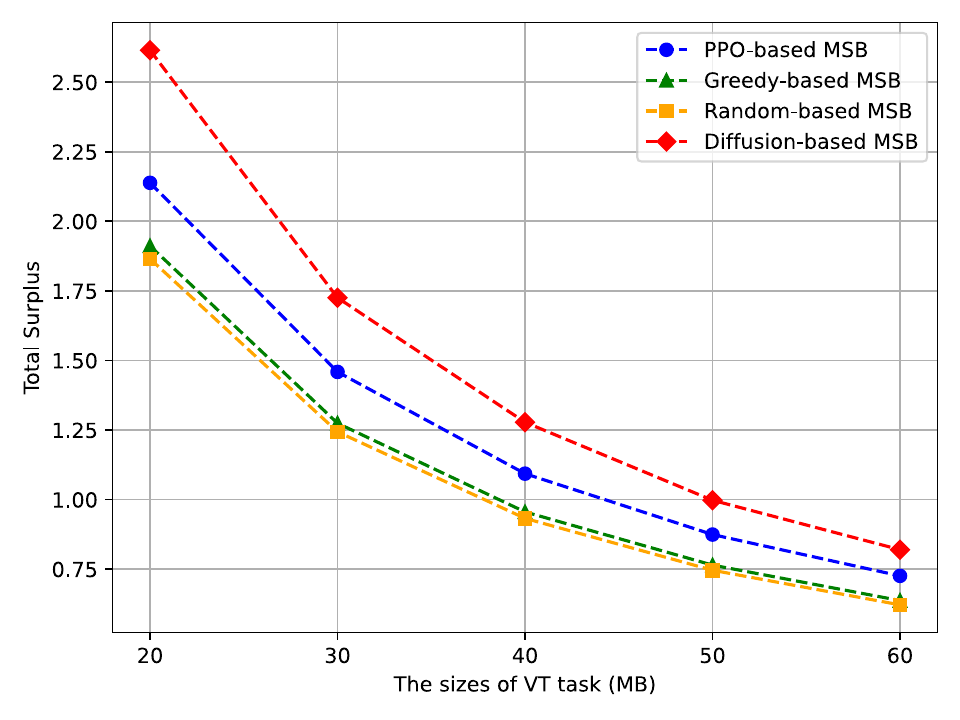}
    }
    \hspace{0.02\textwidth} 
    \subfloat[\centering Average delay versus Different sizes of VT task\label{delay}]{
        \includegraphics[width=0.4\textwidth]{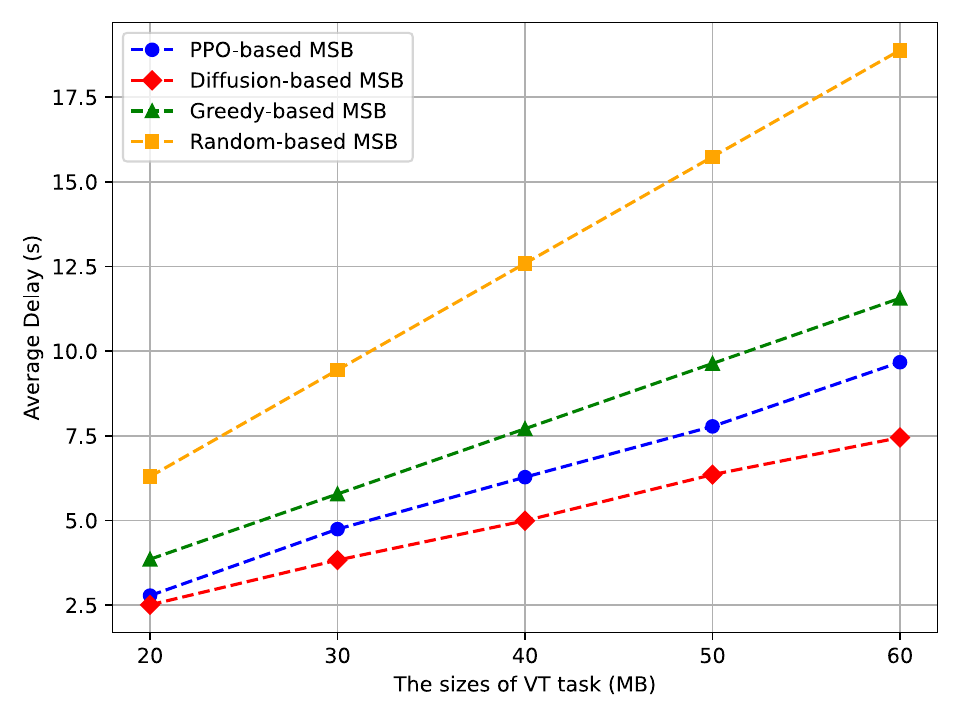}
    }

    \caption{Comparison of total surplus under different environment settings and different algorithms.}
    \label{fig_algorithm}
\end{figure}
\subsection{Convergence Analysis}
We analyze the convergence performance of the proposed DMSB auction compared to several baseline strategies, including PPO, Greedy, Random, and theoretical auction methods, as shown in Fig.~\ref{fig_cove_all}. The proposed DMSB auction demonstrates a rapid increase in total surplus at the beginning of training and stabilizes after approximately 500,000 environment timesteps, achieving a significantly higher total surplus than the baseline methods. Compared to PPO-based MSB auction, the DMSB auction achieves 13\% higher surplus. Moreover, the total surplus of our proposed DMSB auction is significantly higher those that of the Greedy-based MSB auction and random-based MSB auction, by 45\% and 63\% respectively, highlighting the effectiveness of dynamic optimization in the auction environment. In terms of auction baselines, the DMSB auction outperforms all traditional methods, with 30\% higher total surplus than the SPA, 27\% higher than the Myopic MSB auction, and 23\% higher than the Optimal MSB auction.

We further analyze the convergence of the DMSB auction under surplus compared with different numbers of resource providers. As illustrated in Figs.~\ref{fig:3_1_convergence} -~\ref{fig:9_1_convergence}, the DMSB auction consistently achieves a higher total surplus compared to traditional auction baselines, highlighting its robustness in convergence under different resource providers. Initially, as shown in Fig.~\ref{fig:3_1_convergence}, with fewer ground BSs (i.e., 3 ground BSs and 1 UAV), the UAV achieves a higher surplus, indicating it acts as the primary resource provider to handle most of the VT tasks for vehicular users. However, as the number of ground BSs increases, as seen in Figs.~\ref{fig:5_1_convergence} -~\ref{fig:9_1_convergence}, the surplus of ground BSs grows steadily and surpasses the UAV. Despite this shift, the UAV maintains a stable surplus, though may lower than ground BSs.
\subsection{Performance Evaluation For the Proposed Auction} We evaluate the performance of the proposed DMSB auction across different auction environments, focusing on bandwidth, seller amounts, and VT task sizes, as shown in Fig.~\ref{fig:overall-revenue-plots}. In particular, Fig.~\ref{fig:revenue_bandwidth} shows that as seller bandwidths increase, the total surplus of all auction mechanisms also rises, with the proposed DMSB auction consistently outperforming SPA, Myopic MSB, and Optimal MSB. It can be seen that ground BSs contribute the most to the total surplus. In contrast, the UAV surplus remains stable, highlighting the positive correlation between bandwidth and surplus generation, particularly for ground BSs to process VT tasks. Fig.~\ref{fig:revenue_seller} presents the total surplus as the number of ground BSs increases. Although there is a slight decline in total surplus with more ground BSs, the DMSB auction consistently achieves a higher total surplus than other auction mechanisms. The marginal utility of adding more BSs decreases, leading to a lower total surplus, but the DMSB auction remains robust, demonstrating superior performance even with more sellers. In Fig.~\ref{fig:revenue_vt_size}, as the size of VT tasks increases, the total surplus decreases across all auctions due to higher computing and communication demands. However, the DMSB auction maintains a consistently higher total surplus than those of the SPA, Myopic MSB, and Optimal MSB, illustrating its robustness in handling greater computing demands.

We also evaluate the performance of the Diffusion-based RL algorithm compared with PPO, Greedy, and Random, as shown in Fig.~\ref{fig_algorithm}. From Figs.~\ref{bandwidth} -~\ref{VTsize}, it can be observed that as seller bandwidths, seller amounts, and VT task sizes increase, the total surplus achieved by the proposed DMSB auction consistently outperforms PPO, Greedy, and Random algorithms. It demonstrates the advantage of the Diffusion-based RL algorithm in dynamically adjusting the price scaling factor $\rho$, leading to more efficient resource allocation and higher total surplus. In addition, Fig.~\ref{delay} shows as the size of VT tasks increases, the average delay rises across all auction mechanisms. However, the DMSB auction consistently achieves the lowest delay compared to PPO-based, Greedy-based, and Random-based MSB auctions.

\section{Conclusion}\label{conclusion}
In this paper, we have proposed a novel DMSB auction for efficient resource allocation of large AI model-based VT tasks in 6G-enabled Vehicular Metaverses. The proposed auction addresses the challenges of dynamic resource demands and information asymmetry between UAVs and ground BSs in the auction process by incorporating latency and task accuracy as common values and match values, respectively. In addition, we have devised a Diffusion-based RL algorithm to dynamically adjust the auction's price scaling factor, thereby optimizing the allocation of large AI model-based VT tasks and maximizing the total surplus of resource providers while minimizing latency. The simulation results have demonstrated the superior performance of the proposed method, achieving improved resource distribution, reduced latency, and enhanced service quality compared to traditional auction baselines.






\bibliographystyle{IEEEtran}
\bibliography{main}

\begin{thebibliography}{10}
\providecommand{\url}[1]{#1}
\csname url@samestyle\endcsname
\providecommand{\newblock}{\relax}
\providecommand{\bibinfo}[2]{#2}
\providecommand{\BIBentrySTDinterwordspacing}{\spaceskip=0pt\relax}
\providecommand{\BIBentryALTinterwordstretchfactor}{4}
\providecommand{\BIBentryALTinterwordspacing}{\spaceskip=\fontdimen2\font plus
\BIBentryALTinterwordstretchfactor\fontdimen3\font minus \fontdimen4\font\relax}
\providecommand{\BIBforeignlanguage}[2]{{%
\expandafter\ifx\csname l@#1\endcsname\relax
\typeout{** WARNING: IEEEtran.bst: No hyphenation pattern has been}%
\typeout{** loaded for the language `#1'. Using the pattern for}%
\typeout{** the default language instead.}%
\else
\language=\csname l@#1\endcsname
\fi
#2}}
\providecommand{\BIBdecl}{\relax}
\BIBdecl

\bibitem{zhong2023blockchain}
Y.~Zhong, J.~Wen, J.~Zhang, J.~Kang, Y.~Jiang, Y.~Zhang, Y.~Cheng, and Y.~Tong, ``Blockchain-assisted twin migration for vehicular metaverses: A game theory approach,'' \emph{Transactions on Emerging Telecommunications Technologies}, vol.~34, no.~12, p. e4856, 2023.

\bibitem{tong2024diffusion}
\BIBentryALTinterwordspacing
Y.~Tong, J.~Kang, J.~Chen, M.~Xu, G.~Li, W.~Zhang, and X.~Yan, ``Diffusion-based reinforcement learning for dynamic uav-assisted vehicle twins migration in vehicular metaverses,'' 2024. [Online]. Available: \url{https://arxiv.org/abs/2406.05422}
\BIBentrySTDinterwordspacing

\bibitem{luo2023privacy}
X.~Luo, J.~Wen, J.~Kang, J.~Nie, Z.~Xiong, Y.~Zhang, Z.~Yang, and S.~Xie, ``Privacy attacks and defenses for digital twin migrations in vehicular metaverses,'' \emph{IEEE Network}, 2023.

\bibitem{chen2023multiple}
J.~Chen, J.~Nie, M.~Xu, L.~Lyu, Z.~Xiong, J.~Kang, Y.~Tong, and W.~Jiang, ``Multiple-agent deep reinforcement learning for avatar migration in vehicular metaverses,'' in \emph{Companion Proceedings of the ACM Web Conference 2023}, 2023, pp. 1258--1265.

\bibitem{10401029}
D.~S. Sarwatt, Y.~Lin, J.~Ding, Y.~Sun, and H.~Ning, ``Metaverse for intelligent transportation systems (its): A comprehensive review of technologies, applications, implications, challenges and future directions,'' \emph{IEEE Transactions on Intelligent Transportation Systems}, vol.~25, no.~7, pp. 6290--6308, 2024.

\bibitem{pang2021uav}
X.~Pang, M.~Sheng, N.~Zhao, J.~Tang, D.~Niyato, and K.-K. Wong, ``When uav meets irs: Expanding air-ground networks via passive reflection,'' \emph{IEEE Wireless Communications}, vol.~28, no.~5, pp. 164--170, 2021.

\bibitem{kang2024uav}
J.~Kang, J.~Chen, M.~Xu, Z.~Xiong, Y.~Jiao, L.~Han, D.~Niyato, Y.~Tong, and S.~Xie, ``Uav-assisted dynamic avatar task migration for vehicular metaverse services: A multi-agent deep reinforcement learning approach,'' \emph{IEEE/CAA Journal of Automatica Sinica}, vol.~11, no.~2, pp. 430--445, 2024.

\bibitem{dong2021uavs}
C.~Dong, Y.~Shen, Y.~Qu, K.~Wang, J.~Zheng, Q.~Wu, and F.~Wu, ``Uavs as an intelligent service: Boosting edge intelligence for air-ground integrated networks,'' \emph{IEEE Network}, vol.~35, no.~4, pp. 167--175, 2021.

\bibitem{noor2020survey}
M.~Noor-A-Rahim, Z.~Liu, H.~Lee, G.~M.~N. Ali, D.~Pesch, and P.~Xiao, ``A survey on resource allocation in vehicular networks,'' \emph{IEEE transactions on intelligent transportation systems}, vol.~23, no.~2, pp. 701--721, 2020.

\bibitem{mozaffari2019tutorial}
M.~Mozaffari, W.~Saad, M.~Bennis, Y.-H. Nam, and M.~Debbah, ``A tutorial on uavs for wireless networks: Applications, challenges, and open problems,'' \emph{IEEE communications surveys \& tutorials}, vol.~21, no.~3, pp. 2334--2360, 2019.

\bibitem{zhu2022traffic}
L.~Zhu, M.~M. Karim, K.~Sharif, C.~Xu, and F.~Li, ``Traffic flow optimization for uavs in multi-layer information-centric software-defined fanet,'' \emph{IEEE Transactions on Vehicular Technology}, vol.~72, no.~2, pp. 2453--2467, 2022.

\bibitem{arnosti2016adverse}
N.~Arnosti, M.~Beck, and P.~Milgrom, ``Adverse selection and auction design for internet display advertising,'' \emph{American Economic Review}, vol. 106, no.~10, pp. 2852--2866, 2016.

\bibitem{xu2024cached}
M.~Xu, D.~Niyato, H.~Zhang, J.~Kang, Z.~Xiong, S.~Mao, and Z.~Han, ``Cached model-as-a-resource: Provisioning large language model agents for edge intelligence in space-air-ground integrated networks,'' \emph{arXiv preprint arXiv:2403.05826}, 2024.

\bibitem{du2024enhancing}
H.~Du, R.~Zhang, Y.~Liu, J.~Wang, Y.~Lin, Z.~Li, D.~Niyato, J.~Kang, Z.~Xiong, S.~Cui \emph{et~al.}, ``Enhancing deep reinforcement learning: A tutorial on generative diffusion models in network optimization,'' \emph{IEEE Communications Surveys \& Tutorials}, 2024.

\bibitem{10419041}
H.~Cao, C.~Tan, Z.~Gao, Y.~Xu, G.~Chen, P.-A. Heng, and S.~Z. Li, ``A survey on generative diffusion models,'' \emph{IEEE Transactions on Knowledge and Data Engineering}, vol.~36, no.~7, pp. 2814--2830, 2024.

\bibitem{stephenson1994snow}
N.~Stephenson, \emph{Snow crash}.\hskip 1em plus 0.5em minus 0.4em\relax Penguin UK, 1994.

\bibitem{mystakidis2022metaverse}
S.~Mystakidis, ``Metaverse,'' \emph{Encyclopedia}, vol.~2, no.~1, pp. 486--497, 2022.

\bibitem{wang2022survey}
Y.~Wang, Z.~Su, N.~Zhang, R.~Xing, D.~Liu, T.~H. Luan, and X.~Shen, ``A survey on metaverse: Fundamentals, security, and privacy,'' \emph{IEEE Communications Surveys \& Tutorials}, vol.~25, no.~1, pp. 319--352, 2022.

\bibitem{zhou2022vetaverse}
P.~Zhou, J.~Zhu, Y.~Wang, Y.~Lu, Z.~Wei, H.~Shi, Y.~Ding, Y.~Gao, Q.~Huang, Y.~Shi \emph{et~al.}, ``Vetaverse: A survey on the intersection of metaverse, vehicles, and transportation systems,'' \emph{arXiv preprint arXiv:2210.15109}, 2022.

\bibitem{xu2023generative}
M.~Xu, D.~Niyato, H.~Zhang, J.~Kang, Z.~Xiong, S.~Mao, and Z.~Han, ``Generative {AI}-empowered effective physical-virtual synchronization in the vehicular metaverse,'' in \emph{2023 IEEE International Conference on Metaverse Computing, Networking and Applications (MetaCom)}.\hskip 1em plus 0.5em minus 0.4em\relax IEEE, 2023, pp. 607--611.

\bibitem{xu2023epvisa}
M.~Xu, D.~Niyato, B.~Wright, H.~Zhang, J.~Kang, Z.~Xiong, S.~Mao, and Z.~Han, ``Epvisa: Efficient auction design for real-time physical-virtual synchronization in the human-centric metaverse,'' \emph{IEEE Journal on Selected Areas in Communications}, vol.~42, no.~3, pp. 694--709, 2024.

\bibitem{juarez2021digital}
M.~G. Juarez, V.~J. Botti, and A.~S. Giret, ``Digital twins: Review and challenges,'' \emph{Journal of Computing and Information Science in Engineering}, vol.~21, no.~3, p. 030802, 2021.

\bibitem{du2023yolo}
B.~Du, H.~Du, H.~Liu, D.~Niyato, P.~Xin, J.~Yu, M.~Qi, and Y.~Tang, ``{YOLO}-based semantic communication with generative {AI}-aided resource allocation for digital twins construction,'' \emph{IEEE Internet of Things Journal}, vol.~11, no.~5, pp. 7664--7678, 2024.

\bibitem{alkhoori2024latency}
F.~A. AlKhoori, L.~U. Khan, M.~Guizani, and M.~Takac, ``Latency-aware placement of vehicular metaverses using virtual network functions,'' \emph{Simulation Modelling Practice and Theory}, vol. 133, p. 102899, 2024.

\bibitem{wen2023task}
J.~Wen, J.~Kang, Z.~Xiong, Y.~Zhang, H.~Du, Y.~Jiao, and D.~Niyato, ``Task freshness-aware incentive mechanism for vehicle twin migration in vehicular metaverses,'' pp. 481--487, June 2023.

\bibitem{10373683}
X.~Zhang, Q.~Zhu, and H.~V. Poor, ``Neyman-pearson criterion driven nfv-sdn architectures and optimal resource-allocations for statistical-qos based murllc over next- generation metaverse mobile networks using fbc,'' \emph{IEEE Journal on Selected Areas in Communications}, vol.~42, no.~3, pp. 570--587, 2024.

\bibitem{9838736}
M.~Xu, D.~Niyato, J.~Kang, Z.~Xiong, C.~Miao, and D.~I. Kim, ``Wireless edge-empowered metaverse: A learning-based incentive mechanism for virtual reality,'' in \emph{ICC 2022 - IEEE International Conference on Communications}, 2022, pp. 5220--5225.

\bibitem{10225312}
N.~C. Luong, L.~K. Chau, N.~D.~D. Anh, N.~H. Sang, S.~Feng, V.-D. Nguyen, D.~Niyato, and D.~In~Kim, ``Optimal auction for effective energy management in uav-assisted vehicular metaverse synchronization systems,'' \emph{IEEE Transactions on Vehicular Technology}, vol.~73, no.~1, pp. 1207--1222, 2024.

\bibitem{san2021noise}
R.~San-Roman, E.~Nachmani, and L.~Wolf, ``Noise estimation for generative diffusion models,'' \emph{arXiv preprint arXiv:2104.02600}, 2021.

\bibitem{tong2024multi}
Y.~Tong, J.~Chen, M.~Xu, J.~Kang, Z.~Xiong, D.~Niyato, C.~Yuen, and Z.~Han, ``Multi-attribute auction-based resource allocation for twins migration in vehicular metaverses: A gpt-based drl approach,'' \emph{arXiv preprint arXiv:2406.05418}, 2024.

\bibitem{10.1145/3637528.3671526}
\BIBentryALTinterwordspacing
J.~Guo, Y.~Huo, Z.~Zhang, T.~Wang, C.~Yu, J.~Xu, B.~Zheng, and Y.~Zhang, ``Generative auto-bidding via conditional diffusion modeling,'' in \emph{Proceedings of the 30th ACM SIGKDD Conference on Knowledge Discovery and Data Mining}.\hskip 1em plus 0.5em minus 0.4em\relax New York, NY, USA: Association for Computing Machinery, 2024, p. 5038–5049. [Online]. Available: \url{https://doi.org/10.1145/3637528.3671526}
\BIBentrySTDinterwordspacing

\bibitem{singh2022ai}
A.~Singh, S.~C. Satapathy, A.~Roy, and A.~Gutub, ``Ai-based mobile edge computing for iot: Applications, challenges, and future scope,'' \emph{Arabian Journal for Science and Engineering}, vol.~47, no.~8, pp. 9801--9831, 2022.

\bibitem{xu2023reconfiguring}
J.~Xu, C.~Yuen, C.~Huang, N.~Ul~Hassan, G.~C. Alexandropoulos, M.~Di~Renzo, and M.~Debbah, ``Reconfiguring wireless environments via intelligent surfaces for 6g: reflection, modulation, and security,'' \emph{Science China Information Sciences}, vol.~66, no.~3, pp. 1--130\,304, February 2023.

\bibitem{10515202}
L.~U. Khan, Z.~Han, D.~Niyato, M.~Guizani, and C.~S. Hong, ``Metaverse for wireless systems: Vision, enablers, architecture, and future directions,'' \emph{IEEE Wireless Communications}, vol.~31, no.~4, pp. 245--251, 2024.

\bibitem{du2024diffusion}
H.~Du, Z.~Li, D.~Niyato, J.~Kang, Z.~Xiong, H.~Huang, and S.~Mao, ``Diffusion-based reinforcement learning for edge-enabled ai-generated content services,'' \emph{IEEE Transactions on Mobile Computing}, 2024.

\bibitem{milgrom2021auction}
P.~Milgrom, ``Auction research evolving: Theorems and market designs,'' \emph{American Economic Review}, vol. 111, no.~5, pp. 1383--1405, 2021.

\end{thebibliography}


\end{document}